\documentclass[preprint,12pt]{elsarticle}

\usepackage{amsmath,amssymb,amsfonts,amsthm,mathtools}
\usepackage{booktabs}
\usepackage{array}
\usepackage{graphicx}
\usepackage{tikz}
\usepackage{pgfplots}
\usepackage[hidelinks]{hyperref}
\pgfplotsset{compat=1.18}

\journal{International Journal of Approximate Reasoning}

\newtheorem{theorem}{Theorem}
\newtheorem{proposition}{Proposition}
\newtheorem{corollary}{Corollary}
\newtheorem{remark}{Remark}
\newtheorem{assumption}{Assumption}
\newtheorem{lemma}{Lemma}

\newcommand{\Sset}{\mathcal S}
\newcommand{\Lset}{\mathcal L}
\newcommand{\Pset}{\mathcal P}
\newcommand{\Kset}{\mathcal K}
\newcommand{\Aset}{\mathcal A}
\newcommand{\R}{\mathbb R}
\newcommand{\E}{\mathbb E}
\newcommand{\Prob}{\mathbb P}
\newcommand{\lowerT}{\underline{\mathcal T}}
\newcommand{\upperT}{\overline{\mathcal T}}
\newcommand{\lowerV}{\underline V}
\newcommand{\upperV}{\overline V}
\newcommand{\lowerE}{\underline E}
\newcommand{\upperE}{\overline E}
\newcommand{\NMB}{\mathrm{NMB}}

\begin{document}

\begin{frontmatter}

\title{Imprecise Transition Matrices for Markov Cohort Models: Lower and Upper Expectations with a Practical Health-Economic Application}

\author[medtronic,brown]{Rowan Iskandar\corref{cor1}}
\ead{rowan.iskandar@gmail.com}
\cortext[cor1]{Corresponding author.}
\address[medtronic]{Medtronic Trading S\`{a}rl, Tolochenaz, VD, Switzerland}
\address[brown]{Center for Evidence Synthesis in Health, Department of Health Services, Policy, and Practice, Brown University, Providence, RI, USA}

\begin{abstract}
In applied health research, Markov cohort models are built on a precisely specified transition probability matrix. However, in many applications the available evidence---transition counts, structural constraints, and treatment-effect data---identifies a set of admissible matrices rather than one uniquely justified matrix. This paper formulates an imprecise-probability extension in which inference yields lower and upper expectations over an evidence-compatible set of precise Markov cohort models. The contribution differs from existing imprecise Markov-chain work by focusing on finite-horizon cohort trajectories, additive accumulated outcomes, and transition matrices constructed from empirical transition counts. Under non-empty compact separately specified outgoing-row sets, the lower and upper accumulated outcomes are computed exactly by Bellman-style lower and upper transition operators. We prove the envelope theorem, reduction to the classical model, coherence properties of the lower transition operator, and algebraic conditions under which a single selected matrix yields a non-robust decision. We then show how multinomial transition counts induce admissible matrix sets through the Imprecise Dirichlet Model. A real-world cost-effectiveness example of patent foramen ovale closure after cryptogenic stroke illustrates the practical consequence: the empirical transition matrix slightly favors closure, whereas the imprecise analysis yields an incremental net monetary benefit interval crossing zero. The method provides both a rigorous lower-expectation formulation and a practical diagnostic for decisions that depend on transition probabilities not fully resolved by the evidence.
\end{abstract}


\end{frontmatter}

\section{Introduction}
\label{sec:introduction}

In applied research in health and medicine, Markov cohort models are often used to simulate the prognosis of a patient or a cohort of patients following an intervention in the absence of long-term data \cite{BeckPauker1983,SonnenbergBeck1993,Iskandar2018}. 
The prognosis is represented as a stochastic transition model over a finite set of health states.
These transitions are governed by transition probability matrices, which are the central numerical objects in such models. 
Once a matrix has been estimated, the cohort trajectories and all outcomes are precisely determined. 
Treating the matrix as fixed is appropriate when the transition probabilities are known or can be estimated with a high degree of certainty. 
However, it is less satisfactory when the matrix is estimated from sparse transition counts, heterogeneous evidence, and conflicting expert opinions. 
In those settings, several transition matrices may be compatible with the evidence.
This paper treats that situation as an imprecise-probability problem. Instead of assigning a probability distribution over possible transition matrices, the analyst specifies the set of matrices that remain compatible with the evidence and structural assumptions. 
Each element of this set can be viewed as a standard transition probability matrix defining a precise Markov model. 
Given this set of matrices, we can compute 
the smallest and largest expected outcomes. Computing these bounds is the lower-envelope view used in imprecise probability \cite{Walley1991,Augustin2014,TroffaesDeCooman2014} and in imprecise Markov models \cite{DeCoomanHermansQuaeghebeur2009,HermansSkulj2014,KrakDeBockSiebes2017}. 
The lower-envelope view should not be confused with a Bayesian average over a prior distribution: all admissible precise models are retained on equal footing.
Other related ideas include imprecise Markov chains and lower-transition-operator methods. 
Existing continuous-time work treats an imprecise model as a set of stochastic processes, computes lower expectations as tight lower bounds over that set, and develops lower transition operators as representational and computational tools for those lower expectations \cite{DeCoomanHermansQuaeghebeur2009,Skulj2015,KrakDeBockSiebes2017}. 

Discrete-time imprecise Markov chains have been developed in the same spirit, with global models and lower and upper expectations of functionals such as time averages and accumulated quantities \cite{DeCoomanDeBockLopatatzidis2016,TJoensDeBock2021}. Another related literature is robust Markov decision processes, in which dynamic programming over a \emph{rectangular} ambiguity set---one specified independently across states---yields an exact Bellman recursion \cite{NilimElGhaoui2005,Iyengar2005,Wiesemann2013}.

In this paper, we extend these ideas to finite-horizon cohort trajectories with additive outcomes, using the same state-vector and transition-matrix objects as cohort models.
In this discrete-time setting, the quantities of interest are accumulated outcomes, in addition to the probabilities of occupying a state at a future time point. 
Each such outcome is calculated by summing, over the time horizon, the outcome contribution of each state weighted by the fraction of the cohort occupying it, and may represent survival, cost, quality-adjusted survival, or net monetary benefit. This accumulated-outcome calculation is the formulation used in cohort models and cost-effectiveness analysis, where the cohort is repeatedly redistributed across states according to the transition matrix and outcomes are accumulated over the time horizon of the analyses \cite{SonnenbergBeck1993,Iskandar2018}.

The contributions of this paper are as follows. First, we explicate an imprecise Markov cohort model as a set of admissible transition-matrix sequences and give its lower- and upper-expectation interpretation for additive cohort outcomes, including outcomes that depend on transitions as well as on occupied states. Second, under non-empty compact separately specified outgoing-row sets, we derive exact lower and upper Bellman recursions, prove the coherence of the associated lower transition operators, and show that the resulting accumulated-outcome functional is itself a coherent lower prevision over the admissible set. Third, we establish a performance-difference identity that exactly captures the gap between any two admissible models, and use it to obtain an algebraic diagnostic for determining when the result from a single selected matrix is invariant over the admissible set and when that matrix yields a non-robust decision. Fourth, we present a practical construction of admissible transition matrices from multinomial transition counts using the Imprecise Dirichlet Model (IDM), show that the induced row set coincides exactly with the IDM credal set, and demonstrate the method with a real-world example from a cost-effectiveness analysis.

\section{Preliminaries and notation}
\label{sec:preliminaries}

Let
\[
\Sset=\{S_1,\ldots,S_s\}
\]
be a finite set of mutually exclusive states, where \(s\) is the number of states. We denote by
\[
\Lset(\Sset)=\{f:\Sset\to\R\}
\]
the set of all real-valued functions on \(\Sset\). Because \(\Sset\) is finite, elements of \(\Lset(\Sset)\) will be written as column vectors \(f=(f(1),\ldots,f(s))^\top\), where \(f(i)=f(S_i)\).

A transition probability matrix is an \(s\times s\) matrix \(P=(p_{ij})_{i,j=1}^s\) satisfying \(p_{ij}\ge0\) and \(\sum_{j=1}^s p_{ij}=1\) for every \(i=1,\ldots,s\). Its \(i\)th row is denoted by \(P_{i\cdot}\). If \(X_t\) is the state occupied by one individual at time step \(t\), then \(p_{ij,t}=\Prob(X_{t+1}=S_j\mid X_t=S_i)\).

A cohort vector is a row vector
\[
m_t=(m_{t,1},\ldots,m_{t,s}),\qquad m_{t,j}\ge0,\qquad \sum_{j=1}^s m_{t,j}=1.
\]
It represents the distribution of \(X_t\): \(m_{t,j}=\Prob(X_t=S_j)\). The deterministic cohort trajectory \(m_0,\ldots,m_T\) is therefore the sequence of marginal distributions of the stochastic process.

For \(t=0,\ldots,T\), let
\[
g_t\in\Lset(\Sset)
\]
be an outcome-contribution function. The number \(g_t(i)\) is the amount of the chosen outcome assigned to a person in state \(S_i\) at time step \(t\). Depending on the application, \(g_t\) may represent survival, utility, cost, net monetary benefit, or an indicator for occupying a target state. Because \(g_t\) is allowed to depend on the time step, it also represents discounting and other time-varying valuations: a per-step discount factor \(\beta\in(0,1]\) is encoded by \(g_t(i)=\beta^{t}\tilde g_t(i)\) for an undiscounted contribution \(\tilde g_t\). Some applications additionally assign part of the outcome to the act of transitioning; this transition-based case is treated at the end of Section~\ref{sec:imprecise}.

\section{Classical Markov cohort model}
\label{sec:classical}

Given an initial distribution \(m_0\) and transition matrices \(P_0,\ldots,P_{T-1}\), the classical cohort recursion is
\begin{equation}
m_{t+1}=m_tP_t,\qquad t=0,\ldots,T-1.
\label{eq:forward}
\end{equation}
The accumulated cohort outcome is
\begin{equation}
A(P_{0:T-1})=\sum_{t=0}^{T}m_tg_t,
\label{eq:classical-accumulated}
\end{equation}
where \(P_{0:T-1}\) denotes the sequence \(P_0,\ldots,P_{T-1}\), and \(m_t\) is generated by \eqref{eq:forward}. Equation \eqref{eq:classical-accumulated} is the usual cohort accounting: multiply the fraction of the cohort in each state by the outcome contribution assigned to that state, then add across states and time steps.

The same accumulated outcome can be written backward. Define \(V_t\in\Lset(\Sset)\) by
\[
V_t(i)=\E\left[\sum_{\ell=t}^{T}g_\ell(X_\ell)\,\middle|\,X_t=S_i\right],
\]
where \(\ell\) indexes future time steps. Then \(V_T=g_T\). For \(t<T\), conditioning on the next state gives the Bellman recursion
\begin{equation}
V_T=g_T,\qquad
V_t=g_t+P_tV_{t+1},\qquad t=T-1,\ldots,0.
\label{eq:classical-bellman}
\end{equation}
The forward and backward forms agree:
\[
m_0V_0=A(P_{0:T-1}).
\]
The proof is a direct unrolling of the recursion and is given in the appendix.

The key term in the Bellman recursion \eqref{eq:classical-bellman} is \(P_tV_{t+1}\). Its \(i\)th component is
\[
(P_tV_{t+1})(i)=\sum_{j=1}^s p_{ij,t}V_{t+1}(j).
\]
This is the expected outcome remaining after one transition from current state \(S_i\). If \(P_t\) is fixed, this weighted average is fixed. If several transition matrices remain compatible with the evidence, the same weighted average has a lower and an upper value. That observation is the computational link between Markov cohort models and lower transition operators.

\section{Imprecise transition matrices}
\label{sec:imprecise}

For each time step \(t=0,\ldots,T-1\), let \(\Pset_t\) be a non-empty set of admissible transition probability matrices. A transition-matrix sequence is admissible if
\[
P_{0:T-1}\in\Aset:=\Pset_0\times\cdots\times\Pset_{T-1}.
\]
The imprecise model is the set of precise Markov cohort models generated by all sequences in \(\Aset\). For any additive outcome \(A(P_{0:T-1})\) defined in \eqref{eq:classical-accumulated}, the lower and upper accumulated outcomes are
\begin{equation}
\underline A=\inf_{P_{0:T-1}\in\Aset}A(P_{0:T-1}),
\qquad
\overline A=\sup_{P_{0:T-1}\in\Aset}A(P_{0:T-1}).
\label{eq:envelope}
\end{equation}
These are lower and upper expectations over the admissible set of precise Markov models.

The tractable case used for the formal results is the row-separable, or separately specified, one, stated in the following assumption.

\begin{assumption}[Compact separately specified outgoing rows]
\label{ass:compact-row-sets}
For every \(t=0,\ldots,T-1\) and every \(i=1,\ldots,s\), \(\Kset_{i,t}\) is a non-empty compact set of complete outgoing rows
\[
p_{i,t}=(p_{i1,t},\ldots,p_{is,t}),\qquad p_{ij,t}\ge0,\qquad \sum_{j=1}^s p_{ij,t}=1.
\]
The admissible matrix set at time step \(t\) is
\begin{equation}
\Pset_t=\{P_t:P_{t,i\cdot}\in\Kset_{i,t},\ i=1,\ldots,s\}.
\label{eq:row-separable-set}
\end{equation}
\end{assumption}

The interpretation is that one admissible outgoing row is chosen for each current state and the rows are stacked to form a transition matrix. Compactness guarantees that the row-wise infima and suprema below are attained. Separately specified rows are a structural assumption about the admissible set: the admissible outgoing distribution chosen for one state does not constrain the choice for any other state. This assumption can fail when the rows are linked, through any of three mechanisms that recur below. A \emph{shared parameter} is a single quantity---such as an underlying event rate or a relative risk---that enters more than one row, so those rows cannot be moved independently. A \emph{calibration constraint} requires the matrix to reproduce an observed aggregate, such as overall survival or disease prevalence, which couples entries across rows. A \emph{treatment-effect model} derives one intervention's transitions from another's by applying an effect measure, such as a relative risk or hazard ratio, so the two arms' rows are tied together. Whenever rows are linked in any of these ways, \eqref{eq:row-separable-set} does not hold and the linked matrix set must be optimized as a whole.

\begin{remark}[Time-step-wise versus fixed-matrix admissibility]
\label{rem:rectangular}
The set \(\Aset=\Pset_0\times\cdots\times\Pset_{T-1}\) allows an admissible matrix to be chosen at each time step. This product structure is the standard rectangular interpretation used by the lower and upper Bellman recursions. If the modeling assumption is instead that a single unknown matrix is fixed for all time steps (so that \(\Pset_0=\cdots=\Pset_{T-1}=:\Pset\)), the admissible set is the diagonal \(\{(P,\ldots,P):P\in\Pset\}\subseteq\Aset\). Two consequences follow. First, since the diagonal is contained in the rectangular set, the rectangular envelope is a valid \emph{outer} bound for the fixed-matrix envelope: the rectangular lower (upper) value is no larger (no smaller) than its fixed-matrix counterpart. Second, the two envelopes \emph{coincide} whenever the row-wise minimizer and maximizer selected in \eqref{eq:imprecise-bellman-ijar} can be taken to be the same matrix at every time step. A sufficient condition is that the one-step objective is optimized at a time-invariant row; this holds in particular when each row set is governed by a scalar parameter and the accumulated outcome is monotone in that parameter, in which case the optimizing row is pinned to a parameter endpoint at every step. The example in Section~\ref{sec:pfo} satisfies this condition, so its fixed-matrix endpoint evaluation equals the exact rectangular envelope; when no such time-invariant optimizer exists, the fixed-matrix envelope must be computed by global optimization over \(P\).
\end{remark}

For \(f\in\Lset(\Sset)\), define the lower and upper transition operators
\begin{equation}
(\lowerT_t f)(i)=\inf_{p_{i,t}\in\Kset_{i,t}}\sum_{j=1}^s p_{ij,t}f(j),
\qquad
(\upperT_t f)(i)=\sup_{p_{i,t}\in\Kset_{i,t}}\sum_{j=1}^s p_{ij,t}f(j),
\label{eq:lower-upper-operators-ijar}
\end{equation}
for \(i=1,\ldots,s\). These are lower and upper envelopes of the one-step conditional expectation of \(f(X_{t+1})\), conditional on \(X_t=S_i\).

\begin{proposition}[Coherence properties]
\label{prop:coherence}
Under Assumption~\ref{ass:compact-row-sets}, \(\lowerT_t\) is a coherent lower transition operator in the following finite-state sense. For all \(f,h\in\Lset(\Sset)\), all constants \(c\in\R\), all \(\mu\ge0\), and all states \(S_i\),
\[
(\lowerT_t f)(i)\ge \min_{1\le j\le s} f(j),\qquad
\lowerT_t(f+h)\ge \lowerT_t f+\lowerT_t h,
\]
\[
\lowerT_t(\mu f)=\mu\lowerT_t f,\qquad
\lowerT_t(f+c)=\lowerT_t f+c.
\]
It is also monotone: if \(f\le h\) componentwise, then \(\lowerT_t f\le \lowerT_t h\). The upper operator is conjugate to the lower operator:
\[
\upperT_t f=-\lowerT_t(-f).
\]
The lower bound, super-additivity, and non-negative homogeneity are the defining axioms of a coherent lower prevision in the sense of \cite{Walley1991,TroffaesDeCooman2014}; constant additivity, monotonicity, and the conjugacy relation are consequences. In particular, combining the lower bound with conjugacy gives the two-sided envelope
\[
\min_{1\le j\le s} f(j)\le (\lowerT_t f)(i)\le (\upperT_t f)(i)\le \max_{1\le j\le s} f(j).
\]
\end{proposition}

\begin{theorem}[Envelope Bellman recursion]
\label{thm:bellman-envelope}
Assume Assumption~\ref{ass:compact-row-sets} holds for every \(t=0,\ldots,T-1\). Define
\begin{equation}
\lowerV_T=\upperV_T=g_T,\qquad
\lowerV_t=g_t+\lowerT_t\lowerV_{t+1},\qquad
\upperV_t=g_t+\upperT_t\upperV_{t+1},
\label{eq:imprecise-bellman-ijar}
\end{equation}
for \(t=T-1,\ldots,0\). Then, for every initial cohort vector \(m_0\),
\[
m_0\lowerV_0=\inf_{P_{0:T-1}\in\Aset}A(P_{0:T-1}),
\qquad
m_0\upperV_0=\sup_{P_{0:T-1}\in\Aset}A(P_{0:T-1}).
\]
\end{theorem}

The theorem states that the lower and upper Bellman recursions compute the exact lower and upper expectations in \eqref{eq:envelope}. The proof is given in the appendix. The result is the discrete-time analogue of the lower-transition-operator view used in imprecise Markov models \cite{DeCoomanHermansQuaeghebeur2009,KrakDeBockSiebes2017}. The key simplifying condition is the row-separable structure: the worst outgoing row for one current state can be chosen without constraining the rows for the other current states.

\begin{corollary}[Reduction to the classical model]
\label{cor:singleton}
If every \(\Pset_t\) is a singleton \(\{P_t\}\), then \(\lowerT_t f=\upperT_t f=P_tf\) for all \(f\in\Lset(\Sset)\). Hence \eqref{eq:imprecise-bellman-ijar} reduces to the classical Bellman recursion \eqref{eq:classical-bellman}.
\end{corollary}

\begin{remark}[Coherent lower prevision on accumulated outcomes]
\label{rem:coherent-prevision}
Theorem~\ref{thm:bellman-envelope} has a lower-prevision reading. Fix \(m_0\) and define the functional \(\lowerE(g_{0:T})=m_0\lowerV_0\) on outcome sequences \(g_{0:T}\). Because each one-step operator \(\lowerT_t\) is a coherent lower transition operator (Proposition~\ref{prop:coherence}) and \(m_0\) has non-negative entries summing to one, the composition in \eqref{eq:imprecise-bellman-ijar} preserves the lower bound, super-additivity, and non-negative homogeneity, together with constant additivity and monotonicity. Hence \(\lowerE\) is a coherent lower prevision on \(\Lset(\Sset)^{T+1}\), and \(\upperE=-\lowerE(-\cdot)\) is its conjugate upper prevision. The envelope is therefore not a heuristic sensitivity range but the lower expectation of a coherent imprecise model, consistent with the operator view of imprecise Markov chains \cite{DeCoomanHermansQuaeghebeur2009,KrakDeBockSiebes2017}.
\end{remark}

\paragraph{Transition-based outcome contributions}
\label{sec:transition-rewards}
Some outcomes are incurred when a transition occurs rather than while a state is occupied; an acute event cost charged on entry to an event state is the canonical example. Let \(r_t(i,j)\in\R\) be a one-step transition contribution, interpreted as the outcome assigned to a person who moves from \(S_i\) to \(S_j\) during step \(t\). The accumulated outcome \eqref{eq:classical-accumulated} generalizes to
\begin{equation}
A(P_{0:T-1})=\sum_{t=0}^{T}m_tg_t+\sum_{t=0}^{T-1}\sum_{i=1}^s\sum_{j=1}^s m_{t,i}\,p_{ij,t}\,r_t(i,j),
\label{eq:accumulated-transition}
\end{equation}
and the backward recursion \eqref{eq:classical-bellman} becomes
\begin{equation}
V_T=g_T,\qquad V_t(i)=g_t(i)+\sum_{j=1}^s p_{ij,t}\bigl(r_t(i,j)+V_{t+1}(j)\bigr),\qquad t=T-1,\ldots,0,
\label{eq:bellman-transition}
\end{equation}
for which the forward--backward equivalence \(m_0V_0=A(P_{0:T-1})\) holds verbatim. For the imprecise model, the only change to the envelope recursion \eqref{eq:imprecise-bellman-ijar} is that the operators act on \(r_t(i,\cdot)+\lowerV_{t+1}(\cdot)\) instead of \(\lowerV_{t+1}(\cdot)\):
\begin{equation}
\lowerV_t(i)=g_t(i)+\inf_{p_{i,t}\in\Kset_{i,t}}\sum_{j=1}^s p_{ij,t}\bigl(r_t(i,j)+\lowerV_{t+1}(j)\bigr),
\label{eq:imprecise-bellman-transition}
\end{equation}
and analogously for \(\upperV_t\) with the supremum. Because \(r_t(i,\cdot)\) is fixed and the objective in \eqref{eq:imprecise-bellman-transition} is still linear in the row \(p_{i,t}\), Assumption~\ref{ass:compact-row-sets} and the proof of Theorem~\ref{thm:bellman-envelope} apply unchanged, so \eqref{eq:imprecise-bellman-transition} computes the exact lower (upper) envelope of \eqref{eq:accumulated-transition}. The state-based development of Sections~\ref{sec:classical}--\ref{sec:imprecise} is the special case \(r_t\equiv0\). The performance-difference identity of Section~\ref{sec:bias} extends as well, with the row difference acting on the augmented one-step value \(r_t(i,j)+\widehat V_{t+1}(j)\). The example in Section~\ref{sec:pfo} uses exactly such a contribution for the acute recurrent-stroke cost, so it is an instance of \eqref{eq:accumulated-transition}.

\section{When can a single selected matrix mislead?}
\label{sec:bias}

For any transition-matrix sequence supplied to it, the classical cohort calculation returns the correct accumulated outcome; the difficulty is not the arithmetic but the input. When a single selected sequence is treated as uniquely justified, the outcome it produces---and the decision that outcome supports---can differ from those of other sequences equally compatible with the evidence. This section makes that dependence precise: a performance-difference identity gives the exact gap between any two admissible models, from which we obtain an algebraic condition for when the selected-matrix result is invariant over the admissible set, and a corresponding test for when it yields a non-robust decision.

Let \(\widehat P_{0:T-1}\in\Aset\) be the selected sequence used in a classical point analysis, and let \(\widehat V_t\) be the corresponding backward vector from \eqref{eq:classical-bellman}. For any admissible sequence \(Q_{0:T-1}\in\Aset\), let \(m_t^Q\) be the cohort trajectory generated by \(Q\).

\begin{proposition}[Performance-difference identity]
\label{prop:performance-difference}
For any \(Q_{0:T-1}\in\Aset\),
\begin{equation}
A(Q_{0:T-1})-A(\widehat P_{0:T-1})
=
\sum_{t=0}^{T-1}m_t^Q(Q_t-\widehat P_t)\widehat V_{t+1}.
\label{eq:performance-difference}
\end{equation}
\end{proposition}

Thus the point result is invariant over \(\Aset\) if and only if the right-hand side of \eqref{eq:performance-difference} is zero for every admissible \(Q_{0:T-1}\). Row by row,
\[
m_t^Q(Q_t-\widehat P_t)\widehat V_{t+1}
=
\sum_{i=1}^s m_{t,i}^Q\sum_{j=1}^s (q_{ij,t}-\widehat p_{ij,t})\,\widehat V_{t+1}(j).
\]
Unresolved transition probabilities can affect the accumulated outcome only if the current state is reached, the admissible outgoing row differs from the selected row, and the row difference moves probability mass among destination states with different remaining outcomes. This algebraic condition is the precise version of the practical warning that an empirical zero should not be confused with a structural impossibility.

For decisions, let \(D(P^a,P^b;\lambda)\) be the incremental net monetary benefit (INMB) for intervention \(a\) versus \(b\) at willingness-to-pay threshold \(\lambda\). Define
\[
\underline D=\inf_{(P^a,P^b)\in\Aset_{ab}}D(P^a,P^b;\lambda),
\qquad
\overline D=\sup_{(P^a,P^b)\in\Aset_{ab}}D(P^a,P^b;\lambda),
\]
where \(\Aset_{ab}\) is the admissible joint set for the two interventions. When the interventions are informed by separate evidence, \(\Aset_{ab}=\Aset_a\times\Aset_b\) is a product set; writing \(\NMB_a\) and \(\NMB_b\) for the net monetary benefit of each intervention, so that \(D=\NMB_a-\NMB_b\), the lower endpoint then factorizes as
\[
\underline D=\inf_{P^a\in\Aset_a}\NMB_a-\sup_{P^b\in\Aset_b}\NMB_b,
\]
and analogously \(\overline D=\sup_{P^a}\NMB_a-\inf_{P^b}\NMB_b\), so the two arms are optimized separately by the recursion of Theorem~\ref{thm:bellman-envelope}. If the arms share parameters, \(\Aset_{ab}\) is not a product and \(D\) must be optimized jointly over the linked set.

\begin{proposition}[Decision robustness]
\label{prop:decision-robustness}
Suppose the selected matrices recommend intervention \(a\), so that
\[
D(\widehat P^a,\widehat P^b;\lambda)>0.
\]
Then the recommendation is robust over \(\Aset_{ab}\) if \(\underline D>0\). If \(\underline D\le0\), the admissible set does not support a strict recommendation for \(a\) against \(b\). If \(\underline D<0\), at least one admissible pair of transition-matrix sequences favors \(b\).
\end{proposition}

\section{Constructing admissible matrices from transition counts}
\label{sec:idm}

The previous sections assume that \(\Pset_t\) is specified. This section gives a data-based construction using multinomial transition counts and the IDM.

Fix current state \(S_i\). Let \(n_{ij}\) be the observed number of transitions from \(S_i\) to \(S_j\), let \(n_i=(n_{i1},\ldots,n_{is})\), and let \(N_i=\sum_{j=1}^s n_{ij}\). For the true outgoing row \(p_i=(p_{i1},\ldots,p_{is})\), the usual row likelihood is
\begin{equation}
n_i\mid p_i\sim\mathrm{Multinomial}(N_i,p_i).
\label{eq:multinomial-row}
\end{equation}
A classical empirical row uses \(\widehat p_{ij}=n_{ij}/N_i\) when \(N_i>0\). This estimator is convenient but treats observed zeros as zero transition probabilities.

The IDM uses a class of Dirichlet priors
\[
\{\mathrm{Dirichlet}(\alpha q_1,\ldots,\alpha q_s):q_j>0,\ \sum_{j=1}^s q_j=1\},
\]
where \(\alpha>0\) is a prior strength. After observing \(n_i\), the posterior class is
\[
\{\mathrm{Dirichlet}(n_{i1}+\alpha q_1,\ldots,n_{is}+\alpha q_s):q_j>0,\ \sum_{j=1}^s q_j=1\}.
\]
Optimizing the posterior mean of \(p_{ij}\) over this class gives
\begin{equation}
\underline p_{ij}=\frac{n_{ij}}{N_i+\alpha},
\qquad
\overline p_{ij}=\frac{n_{ij}+\alpha}{N_i+\alpha},
\qquad j=1,\ldots,s.
\label{eq:idm-bounds-ijar}
\end{equation}
The induced IDM row set is
\begin{equation}
\Kset_i^{IDM}
=
\left\{p_i:
p_{ij}\ge0,\ 
\sum_{j=1}^s p_{ij}=1,\ 
\underline p_{ij}\le p_{ij}\le \overline p_{ij},\ j=1,\ldots,s
\right\}.
\label{eq:idm-row-set-ijar}
\end{equation}
Repeating this construction for all current states and applying structural constraints gives the admissible transition matrices used in \eqref{eq:row-separable-set}. The interval endpoints in \eqref{eq:idm-bounds-ijar} are not independent parameters; the simplex condition in \eqref{eq:idm-row-set-ijar} preserves the fact that each row must sum to one. The next lemma records that this set is exactly the IDM credal set, so it satisfies Assumption~\ref{ass:compact-row-sets} and the upper bounds in \eqref{eq:idm-row-set-ijar} are in fact redundant.

\begin{lemma}[Exactness of the IDM row set]
\label{lem:idm-exact}
With \(\underline p_{ij},\overline p_{ij}\) as in \eqref{eq:idm-bounds-ijar}, the row set is
\[
\Kset_i^{IDM}
=\Bigl\{p_i:\ p_{ij}\ge \underline p_{ij}\ (j=1,\ldots,s),\ \textstyle\sum_{j=1}^s p_{ij}=1\Bigr\}
=\operatorname{conv}\{v_1,\ldots,v_s\},
\]
where
\[
v_k=\frac{n_i+\alpha e_k}{N_i+\alpha},\qquad k=1,\ldots,s,
\]
and \(e_k\) is the \(k\)th standard basis vector. Hence \(\Kset_i^{IDM}\) is a non-empty compact polytope (the set of IDM posterior means), and the upper bounds \(p_{ij}\le\overline p_{ij}\) follow from the lower bounds and the simplex constraint.
\end{lemma}

\begin{proof}
If \(\sum_j p_{ij}=1\) and \(p_{ij}\ge\underline p_{ij}\) for all \(j\), then for each \(k\),
\(p_{ik}=1-\sum_{j\ne k}p_{ij}\le 1-\sum_{j\ne k}\underline p_{ij}=1-\tfrac{N_i-n_{ik}}{N_i+\alpha}=\tfrac{n_{ik}+\alpha}{N_i+\alpha}=\overline p_{ik}\),
so the upper bounds are redundant. Setting \(w_k=(p_{ik}-\underline p_{ik})(N_i+\alpha)/\alpha\) gives \(w_k\ge0\), \(\sum_k w_k=\bigl(1-N_i/(N_i+\alpha)\bigr)(N_i+\alpha)/\alpha=1\), and \(p_i=\sum_k w_k v_k\); conversely every such convex combination meets the constraints. Thus the three sets coincide.
\end{proof}

The framework does not depend on the IDM. Any non-empty compact row set \(\Kset_{i,t}\)---for instance one defined by a coherent lower prevision, a credal set from linear probability constraints, interval probabilities, or an \(\epsilon\)-contamination model---can be used in \eqref{eq:lower-upper-operators-ijar} and Theorem~\ref{thm:bellman-envelope}. We use the IDM in the example because it maps transition counts directly to row sets.

\section{Computational algorithm}
\label{sec:computation}

We provide a computational algorithm and the accompanying code as supplementary files. The computation follows directly from Theorem~\ref{thm:bellman-envelope}. For a selected outcome contribution \(g_t\), set \(\lowerV_T=\upperV_T=g_T\). For \(t=T-1,\ldots,0\), compute the lower and upper one-step expectations in \eqref{eq:lower-upper-operators-ijar} and update \(\lowerV_t,\upperV_t\) using \eqref{eq:imprecise-bellman-ijar}. The final interval for the initial cohort is
\[
[m_0\lowerV_0,\ m_0\upperV_0].
\]

For each state \(S_i\), the optimization in \eqref{eq:lower-upper-operators-ijar} is a linear program:
\[
\min_{p_{i,t}\in\Kset_{i,t}}\sum_{j=1}^s p_{ij,t}f(j)
\quad\text{or}\quad
\max_{p_{i,t}\in\Kset_{i,t}}\sum_{j=1}^s p_{ij,t}f(j).
\]
For IDM row sets, the constraints are linear: bounds, non-negativity, and summing to one. Hence the computation is transparent and scales linearly in the number of time steps once the row-level linear programs are solved.

The algorithm is as follows:

\begin{enumerate}
\item Define the states, initial distribution, time horizon, and outcome contributions \(g_t\).
\item Estimate an empirical transition matrix for comparison.
\item Build IDM row sets from transition counts using \eqref{eq:idm-bounds-ijar}--\eqref{eq:idm-row-set-ijar}.
\item Apply the lower and upper Bellman recursion \eqref{eq:imprecise-bellman-ijar}.
\item Report the empirical point result together with the lower and upper outcomes.
\item For decisions, optimize INMB directly rather than optimizing costs and effects separately.
\end{enumerate}


\section{Real-world application: patent foramen ovale closure after cryptogenic stroke}
\label{sec:pfo}

We use a cost-effectiveness example to demonstrate how the proposed method can change a decision. The clinical question is whether patent foramen ovale (PFO) closure is cost-effective compared with medical therapy alone for secondary prevention after cryptogenic stroke. Published economic evaluations have generally found PFO closure cost-effective or cost-saving over longer horizons \cite{Leppert2018PFOCEA}. The specific aim of this application is to examine the effect of sparse recurrent-stroke evidence near a decision boundary. The code for this example is included as supplementary material.
We also conduct an independent verification: rather than evaluating the model at the interval endpoints, we use the full backward recursion of Theorem~\ref{thm:bellman-envelope}, applying the lower and upper transition operators at every time step in the transition-reward form~\eqref{eq:imprecise-bellman-transition} to assign a recurrent-stroke cost to each transition.

\subsection{Data on clinical events}

The example uses recurrent-stroke data from two high-risk PFO trials (Table~\ref{tab:pfo-evidence-ijar}). In CLOSE, no stroke occurred among 238 patients assigned to PFO closure, whereas 14 strokes occurred among 235 patients assigned to antiplatelet therapy \cite{Mas2017CLOSE}. In DEFENSE-PFO, 60 patients were enrolled in each arm; ischemic stroke occurred in 0\% with closure and 10.5\% with medical therapy \cite{Lee2018DEFENSEPFO}. We represent the latter as six events in the medical-therapy arm. Patient-years are approximated as the trial size times the follow-up period.

\begin{table}[!htbp]
\centering
\caption{Transition evidence for recurrent stroke after cryptogenic stroke with high-risk patent foramen ovale (PFO). PY = patient-years.}
\label{tab:pfo-evidence-ijar}
\resizebox{\textwidth}{!}{%
\begin{tabular}{llrrrr}
\toprule
Trial & Strategy & Patients & Follow-up & Events & PY \\
\midrule
CLOSE & PFO closure & 238 & 5.3 & 0 & 1261.4 \\
CLOSE & Medical therapy & 235 & 5.3 & 14 & 1245.5 \\
DEFENSE-PFO & PFO closure & 60 & 2.8 & 0 & 168.0 \\
DEFENSE-PFO & Medical therapy & 60 & 2.8 & 6 & 168.0 \\
\midrule
Pooled & PFO closure & 298 & -- & 0 & 1429.4 \\
Pooled & Medical therapy & 295 & -- & 20 & 1413.5 \\
\bottomrule
\end{tabular}%
}
\end{table}

Let \(p_{\mathrm{stroke}}\) be the one-time-step probability of recurrent stroke. The empirical probability is 0 for PFO closure and 0.01415 for medical therapy. With IDM prior strength \(\alpha=2\), the IDM intervals are
\[
0\le p_{\mathrm{stroke}}^{\mathrm{closure}}\le 0.00140,
\qquad
0.01413\le p_{\mathrm{stroke}}^{\mathrm{medical}}\le 0.01554.
\]
The closure-arm upper bound is small, but it is not zero. This nonzero upper bound is the distinction the empirical matrix hides. Here the patient-years at risk play the role of the multinomial denominator \(N_i\) in \eqref{eq:idm-bounds-ijar}, so \(p_{\mathrm{stroke}}\) is read as a per-time-step rate; for the small magnitudes involved, the difference between a rate and a one-step probability (\(p\approx 1-e^{-\text{rate}}\)) is negligible. Imprecision is placed only on the recurrent-stroke probability, with the background death probabilities held fixed, so the admissible ``no recurrent stroke'' row is the one-parameter subfamily \(\{(1-p-d_0,\,p,\,d_0):p\in[\underline p,\overline p]\}\) of the full IDM simplex of Lemma~\ref{lem:idm-exact}.

\subsection{Economic model}

The model has three states: no recurrent stroke, after recurrent stroke, and dead. Everyone begins in the no-recurrent-stroke state. The horizon is seven years with 3\% discounting. The willingness-to-pay threshold is \(\lambda=\$150{,}000\) per quality-adjusted life-year (QALY), matching the threshold used in the published PFO economic evaluation \cite{Leppert2018PFOCEA}. Table~\ref{tab:pfo-inputs-ijar} gives the compact model inputs.

\begin{table}[!htbp]
\centering
\caption{Inputs for the compact patent foramen ovale (PFO) closure example. QALY = quality-adjusted life-year.}
\label{tab:pfo-inputs-ijar}
\begin{tabular}{p{0.42\textwidth}p{0.48\textwidth}}
\toprule
Model parameters & Value \\
\midrule
States & no recurrent stroke; after recurrent stroke; dead \\
Initial distribution & \(m_0=(1,0,0)\) \\
Horizon & 7 years \\
Discount rate & 3\% per time step \\
Willingness-to-pay threshold & \(\$150{,}000\) per QALY \\
One-time closure cost & \(\$15{,}000\) \\
Annual cost without recurrent stroke & \(\$200\) \\
Acute recurrent-stroke cost & \(\$40{,}000\) \\
Annual post-stroke cost & \(\$8{,}000\) \\
Utility without recurrent stroke & 0.86 \\
Utility after recurrent stroke & 0.65 \\
Background death probability & 0.002 before recurrent stroke; 0.050 after recurrent stroke \\
\bottomrule
\end{tabular}
\end{table}

The empirical matrix sets \(p_{\mathrm{stroke}}^{\mathrm{closure}}=0\). A precise symmetric Dirichlet posterior mean with strength \(\alpha=2\) gives \(p_{\mathrm{stroke}}^{\mathrm{closure}}=0.00070\) and \(p_{\mathrm{stroke}}^{\mathrm{medical}}=0.01484\). The imprecise analysis evaluates all matrices with recurrent-stroke probabilities in the IDM intervals above.

\subsection{Results}

Table~\ref{tab:pfo-results-ijar} reports discounted incremental net monetary benefit results for PFO closure versus medical therapy. The empirical transition matrix slightly favors closure: INMB is \(\$167\). The precise Dirichlet posterior mean gives the same qualitative conclusion. The IDM analysis changes the interpretation. Its INMB interval crosses zero, from \(-\$1{,}387\) to \(\$1{,}618\) (Figure~\ref{fig:pfo-inmb-ijar}). Thus the empirical matrix recommends closure, but the evidence-compatible matrix set does not support a robust recommendation at this threshold. Because INMB is monotone in each arm's single recurrent-stroke probability, the optimizing transition row is pinned to an interval endpoint at every time step; by Remark~\ref{rem:rectangular} the endpoint evaluation reported here coincides with the exact lower and upper envelopes of Theorem~\ref{thm:bellman-envelope}. The supplementary script \texttt{bellman\_envelope.R} confirms this by computing the same interval directly from the lower and upper transition operators \eqref{eq:imprecise-bellman-transition}.

\begin{table}[!htbp]
\centering
\caption{Discounted cost-effectiveness results for patent foramen ovale (PFO) closure versus medical therapy. Incremental net monetary benefit (INMB) is reported at a willingness-to-pay threshold \(\lambda=\$150{,}000\) per quality-adjusted life-year (QALY). IDM = Imprecise Dirichlet Model.}
\label{tab:pfo-results-ijar}
\resizebox{\textwidth}{!}{%
\begin{tabular}{lrrr}
\toprule
Analysis & Incremental QALYs & Incremental cost & INMB \\
\midrule
Empirical transition matrix & 0.0658 & \(\$9{,}705\) & \(\$167\) \\
Precise Dirichlet posterior mean & 0.0656 & \(\$9{,}727\) & \(\$115\) \\
IDM lower bound & 0.0591 & \(\$10{,}251\) & \(-\$1{,}387\) \\
IDM upper bound & 0.0721 & \(\$9{,}203\) & \(\$1{,}618\) \\
\bottomrule
\end{tabular}%
}
\end{table}

\begin{figure}[!htbp]
\centering
\begin{tikzpicture}
\begin{axis}[
    width=0.80\textwidth,
    height=0.28\textwidth,
    xmin=-1800, xmax=1900,
    ymin=0.35, ymax=1.65,
    xlabel={Incremental net monetary benefit for PFO closure},
    ytick={1},
    yticklabels={IDM interval},
    axis y line*=left,
    axis x line*=bottom,
    grid=none,
    tick style={line width=0.8pt},
    axis line style={line width=0.8pt},
    clip=false
]
\addplot[black, line width=0.9pt] coordinates {(0,0.35) (0,1.65)};
\addplot[blue!70!black, line width=3.0pt] coordinates {(-1387,1) (1618,1)};
\addplot[only marks, mark=*, mark size=2.7pt, blue!70!black] coordinates {(167,1)};
\node[anchor=south, font=\small] at (axis cs:167,1.08) {empirical \(=\$167\)};
\node[anchor=north east, font=\small] at (axis cs:-1387,0.92) {\(-\$1{,}387\)};
\node[anchor=north west, font=\small] at (axis cs:1618,0.92) {\(\$1{,}618\)};
\node[anchor=south west, font=\small] at (axis cs:30,1.45) {zero};
\end{axis}
\end{tikzpicture}
\caption{Empirical point incremental net monetary benefit (INMB) and Imprecise Dirichlet Model (IDM) interval. The empirical matrix favors patent foramen ovale (PFO) closure, but the IDM interval crosses zero.}
\label{fig:pfo-inmb-ijar}
\end{figure}
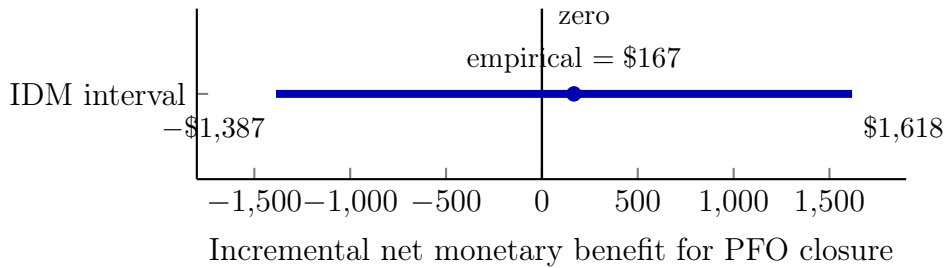

\subsection{Sensitivity to the threshold and the imprecision strength}
\label{sec:pfo-sensitivity}

The non-robust conclusion is not an artifact of the particular threshold or the IDM strength. Holding \(\alpha=2\) and sweeping the willingness-to-pay threshold \(\lambda\) (Table~\ref{tab:pfo-sensitivity}), the empirical point INMB changes sign at \(\lambda\approx\$147{,}500\) per QALY, whereas the lower and upper envelopes change sign at \(\lambda\approx\$173{,}500\) and \(\lambda\approx\$127{,}600\). The IDM interval therefore straddles zero throughout the band \(\$127{,}600<\lambda<\$173{,}500\): PFO closure is a robust choice only above this band and medical therapy only below it. At the \(\$150{,}000\) threshold used in the published evaluation \cite{Leppert2018PFOCEA} the decision is non-robust, and at the lower thresholds \(\$50{,}000\) and \(\$100{,}000\) the imprecise analysis robustly favors medical therapy, even though the empirical point INMB is positive at \(\$150{,}000\).

The interval width scales with the prior strength, as expected. Recomputing at the other conventional IDM strength \(\alpha=1\) gives the INMB interval \([-\$611,\ \$894]\) at \(\lambda=\$150{,}000\), against \([-\$1{,}387,\ \$1{,}618]\) at \(\alpha=2\); both straddle zero, so the non-robustness does not depend on the choice of \(\alpha\).

\begin{table}[!htbp]
\centering
\caption{Discounted incremental net monetary benefit (INMB) interval for patent foramen ovale (PFO) closure versus medical therapy by willingness-to-pay threshold \(\lambda\) (Imprecise Dirichlet Model (IDM) prior strength \(\alpha=2\)). Each interval is the exact lower and upper envelope of Theorem~\ref{thm:bellman-envelope}. \$/QALY = US dollars per quality-adjusted life-year.}
\label{tab:pfo-sensitivity}
\begin{tabular}{rrrl}
\toprule
Threshold \(\lambda\) (\$/QALY) & Lower INMB & Upper INMB & Robustly preferred strategy \\
\midrule
50{,}000  & \(-\$7{,}297\) & \(-\$5{,}596\) & Medical therapy \\
100{,}000 & \(-\$4{,}342\) & \(-\$1{,}989\) & Medical therapy \\
150{,}000 & \(-\$1{,}387\) & \(\$1{,}618\)  & None (non-robust) \\
200{,}000 & \(\$1{,}568\)  & \(\$5{,}225\)  & PFO closure \\
\bottomrule
\end{tabular}
\end{table}

This example is deliberately compact: it does not claim that PFO closure is cost-ineffective in general, but shows that near a decision boundary a zero event count in a finite trial can be decisive in the empirical transition matrix. The imprecise analysis distinguishes an observed zero from a structural zero and makes the resulting non-robustness visible.

\subsection{Behavior under increasing sparsity}
\label{sec:pfo-simulation}

To show that the phenomenon is not tied to the pooled sample size, we study the three analyses as the amount of evidence varies. We fix a ground truth in which closure genuinely lowers the per-step recurrent-stroke probability, from \(0.014\) under medical therapy to \(0.010\) under closure, but not by enough to justify its one-time cost: the true INMB is \(-\$10{,}790\), so the correct decision is medical therapy. For each arm we draw recurrent-stroke counts from a Poisson model with mean \(N\times p\), where \(N\) is the patient-years at risk per arm, and we vary \(N\) from \(25\) to \(3200\). For each simulated data set we form the empirical matrix, the precise symmetric Dirichlet posterior mean (\(\alpha=2\)), and the IDM interval (\(\alpha=2\)), and record the decision each implies at \(\lambda=\$150{,}000\) per QALY. The two point analyses always commit to closure or to medical therapy; the IDM analysis additionally \emph{abstains} when its INMB interval contains zero. Figure~\ref{fig:pfo-sparsity} reports, over \(16{,}000\) replications per sample size, how often each analysis is wrong and how often the IDM abstains.

Two patterns emerge. First, under sparsity the point analyses are frequently confident and wrong: the empirical matrix recommends the inferior arm in up to about a third of samples near \(N=50\), driven by closure-arm sampling zeros that mimic a highly effective treatment, and the precise Dirichlet posterior mean behaves almost identically, so its symmetric prior offers little protection in this regime. Second, the IDM analysis is almost never confidently wrong---its error rate stays below \(5\%\) throughout---because it abstains instead, reporting that the evidence does not resolve the decision. As the sample size grows the sampling zeros disappear, the abstention rate falls to zero, and all three analyses converge to the correct recommendation. The IDM thus trades the point analyses' confident errors for transparent abstentions and commits to the correct decision once the evidence can support it.

This convergence has a simple analytical counterpart. From the IDM bounds \eqref{eq:idm-bounds-ijar}, each admissible row spans an interval of half-width \(\alpha/(N_i+\alpha)=O(\alpha/N_i)\), where \(N_i\) is the patient-years informing that row. Because the accumulated outcome is linear in each row, the INMB interval half-width---and hence the width of the non-robust threshold band of Section~\ref{sec:pfo-sensitivity}---is \(O(\alpha/N)\) in the per-arm patient-years \(N\). The imprecision therefore contracts at rate \(1/N\) as evidence accumulates, which is the formal counterpart of the convergence seen in Figure~\ref{fig:pfo-sparsity}.

\begin{figure}[!htbp]
\centering
\begin{tikzpicture}
\begin{axis}[
    width=0.85\textwidth,
    height=0.42\textwidth,
    xmode=log,
    log basis x=2,
    xmin=20, xmax=4000,
    ymin=-0.03, ymax=1.03,
    xtick={25,50,100,200,400,800,1600,3200},
    xticklabels={25,50,100,200,400,800,1600,3200},
    x tick label style={font=\footnotesize},
    xlabel={Patient-years at risk per arm (\(N\))},
    ylabel={Probability},
    legend pos=north east,
    legend cell align=left,
    legend style={font=\footnotesize, draw=none, fill opacity=0.7, text opacity=1},
    tick style={line width=0.7pt},
    axis line style={line width=0.8pt},
    grid=major,
    grid style={gray!18},
    clip=false
]
\addplot[blue!70!black, line width=1.1pt, mark=*, mark size=1.6pt] coordinates {(25,0.240) (50,0.357) (100,0.222) (200,0.207) (400,0.101) (800,0.030) (1600,0.004) (3200,0.000)};
\addlegendentry{Empirical: wrong decision}
\addplot[red!75!black, line width=1.0pt, densely dashed, mark=square, mark size=1.5pt] coordinates {(25,0.240) (50,0.357) (100,0.222) (200,0.200) (400,0.100) (800,0.030) (1600,0.004) (3200,0.000)};
\addlegendentry{Precise Dirichlet: wrong decision}
\addplot[black, line width=1.1pt, mark=triangle*, mark size=2.0pt] coordinates {(25,0.004) (50,0.023) (100,0.027) (200,0.043) (400,0.027) (800,0.011) (1600,0.001) (3200,0.000)};
\addlegendentry{IDM: wrong decision}
\addplot[black, line width=1.0pt, dotted, mark=o, mark size=1.8pt] coordinates {(25,0.977) (50,0.931) (100,0.704) (200,0.498) (400,0.232) (800,0.065) (1600,0.009) (3200,0.000)};
\addlegendentry{IDM: abstains}
\end{axis}
\end{tikzpicture}
\caption{Decision outcomes under increasing evidence, for a ground truth in which medical therapy is the correct choice (true incremental net monetary benefit (INMB) \(=-\$10{,}790\)). Curves show the probability that each analysis recommends the wrong arm, and the probability that the Imprecise Dirichlet Model (IDM) abstains because its INMB interval contains zero, estimated from \(16{,}000\) simulated trials per sample size at a willingness-to-pay threshold \(\lambda=\$150{,}000\) per quality-adjusted life-year (QALY) and IDM prior strength \(\alpha=2\). The empirical and precise Dirichlet analyses almost coincide. The IDM is almost never confidently wrong: it abstains under sparsity and converges to the correct decision as evidence accumulates.}
\label{fig:pfo-sparsity}
\end{figure}
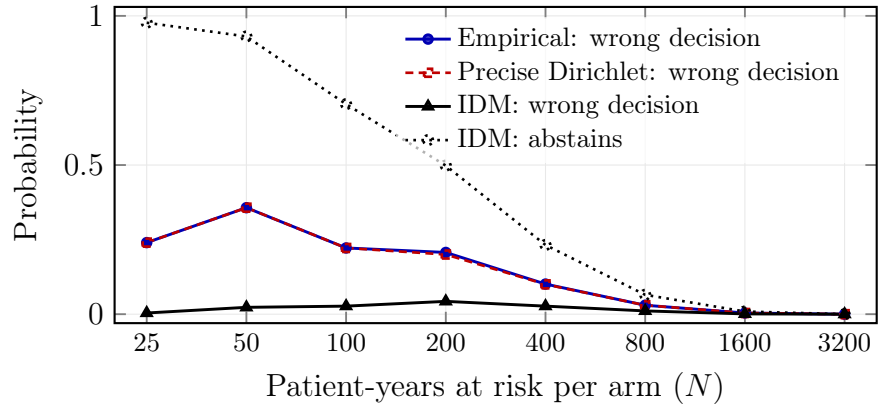

\section{Discussion}
\label{sec:discussion}

The paper develops a lower-expectation formulation of the Markov cohort model with imprecise transition probabilities. Rather than attach a single value to one selected matrix, the method reports the lower and upper expectation of an additive accumulated outcome over an evidence-compatible set of precise cohort models. The classical cohort model is recovered exactly when the admissible set is a singleton at every time step (Corollary~\ref{cor:singleton}); when the admissible set is larger, the lower and upper Bellman recursions of Theorem~\ref{thm:bellman-envelope} return the exact range of the outcome over all admissible models.

\paragraph{Main contribution}
The central result is the envelope theorem (Theorem~\ref{thm:bellman-envelope}), together with the coherence properties of Proposition~\ref{prop:coherence}. Under separately specified rows, the imprecise calculation is not a heuristic sensitivity analysis over hand-picked matrices; it is the exact lower and upper expectation of a coherent imprecise model in the sense of \cite{Walley1991,TroffaesDeCooman2014}, and it coincides with the lower- and upper-transition-operator view of imprecise Markov chains \cite{DeCoomanHermansQuaeghebeur2009,KrakDeBockSiebes2017}. Because the envelope is a coherent lower prevision, interval-valued cohort outcomes stand on the same footing as the point-valued outcomes of the classical model, and the bounds can be read as genuine expectations rather than as the spread of an ad hoc collection of scenarios.

\paragraph{When imprecision changes a decision}
The performance-difference identity of Proposition~\ref{prop:performance-difference} provides an exact diagnostic that complements these bounds. A single selected matrix can mislead only when imprecision both reaches states that the cohort actually occupies and redistributes probability mass among states whose remaining outcomes differ. Imprecision confined to unreachable states, or that moves mass between states with the same continuation value, leaves the decision unchanged. The diagnostic thus identifies in advance where transition uncertainty can and cannot overturn a recommendation, and explains why its consequences are concentrated in a few influential transitions rather than spread uniformly across the matrix.

\paragraph{Practical contribution}
The method is a modest extension of the classical cohort model rather than a departure from it. An analyst can begin from the same transition counts used to build a classical cohort model, construct IDM row sets, run the lower and upper recursions at essentially the cost of a few point evaluations, and report whether the empirical recommendation is robust to the dependence and sampling uncertainty the data leave unresolved. The PFO closure example shows that this is not a formality: the empirical transition matrix and the precise Dirichlet posterior mean both favor closure, yet the evidence-compatible interval straddles the decision boundary, and it continues to do so across the conventional range of willingness-to-pay thresholds and across both conventional IDM strengths. The sparsity experiment of Section~\ref{sec:pfo-simulation} shows that the interval behaves like a calibrated decision rule: it abstains when the evidence is too thin to support either arm, is almost never confidently wrong, and contracts at rate \(O(\alpha/N)\), so that it recovers the point decision as evidence accumulates.

\paragraph{Limitations}
There are two limitations which restrict the scope of the study results, though without the risk of undermining the rigor and utility of the approach in its intended setting. The first is the row-separable structure of Assumption~\ref{ass:compact-row-sets}, which is what makes the envelope recursion tractable: it lets the worst outgoing row for one state be chosen independently of the rows for the other states. Row separability is appropriate when outgoing rows are assessed separately by current state, which is the common case in applied cohort models, where each state's transitions are estimated from its own data. When the rows are instead linked---by a shared parameter, a calibration constraint, or a treatment-effect model---the approach does not break down: the lower and upper expectations remain well defined over the linked set, and, because any linked set is contained in the rectangular product set, the row-by-row envelope is still a valid outer bound for the linked problem (cf.\ Remark~\ref{rem:rectangular}) and can serve as a conservative screen before any costlier joint optimization. Only the row-wise decomposition, not the lower-expectation interpretation, is lost. The second limitation is that the IDM is only one way to convert evidence into an admissible set. Here too the core results are unaffected: the envelope theorem, the coherence of the lower and upper operators, and the performance-difference diagnostic hold for any non-empty compact row sets, so likelihood-based confidence regions, expert-elicited bounds, an imprecise treatment of Bayesian credible sets, or calibration-defined sets all plug into the same recursion. The IDM is simply the construction that matches multinomial transition counts; only its closed-form bounds and the \(O(\alpha/N)\) contraction rate are specific to that setting.

\paragraph{Future directions}
Several extensions follow naturally. Linked transition-matrix sets, in which rows co-vary through shared parameters or calibration constraints, would widen applicability at the cost of row-wise tractability and call for dedicated optimization algorithms, particularly for large state spaces where enumerating rows is impractical. A second direction is to combine imprecise transition matrices with parameter uncertainty analysis. The two answer different questions and are complementary rather than competing: imprecise transition matrices answer a robustness question---whether a recommendation survives every evidence-compatible dependence structure---whereas parameter uncertainty analysis answers a distributional question about the spread of outcomes under an assumed joint distribution of inputs. Reporting both a robustness interval and a distributional summary would give decision makers a fuller account of the uncertainty surrounding a transition-probability-driven decision.

\section*{CRediT authorship contribution statement}
Rowan Iskandar: Conceptualization, Methodology, Formal analysis, Software, Writing--original draft, Writing--review and editing.

\section*{Declaration of competing interest}
The author declares that he has no known competing financial interests or personal relationships that could have appeared to influence the work reported in this paper.

\section*{Funding}
No specific funding was received for this methodological work.

\section*{Data and code availability}
The case-study calculations are reproducible from the R scripts and Quarto guide supplied as supplementary files.

\section*{Acknowledgements}
The author thanks colleagues and reviewers who provided feedback on earlier versions of this manuscript.

\section*{Declaration of generative AI and AI-assisted technologies in the writing process}
During the preparation of this work the author used Claude (a large language model developed by Anthropic) in order to improve the language, readability, and formatting of portions of the manuscript. After using this tool, the author reviewed and edited the content as needed and takes full responsibility for the content of the publication.

\appendix

\section{Proofs}
\label{app:proofs}

\subsection{Forward--backward equivalence}

Unroll \eqref{eq:classical-bellman}. Since \(V_T=g_T\),
\[
V_0=g_0+P_0g_1+P_0P_1g_2+\cdots+P_0P_1\cdots P_{T-1}g_T.
\]
Multiplying by \(m_0\) gives
\[
m_0V_0
=
m_0g_0+m_0P_0g_1+\cdots+m_0P_0P_1\cdots P_{T-1}g_T.
\]
Because \(m_t=m_0P_0\cdots P_{t-1}\), with the empty product interpreted as the identity for \(t=0\), the right-hand side is \(\sum_{t=0}^{T}m_tg_t=A(P_{0:T-1})\).

\subsection{Proof of Proposition~\ref{prop:coherence}}

Fix \(t\) and \(i\). For any \(p_{i,t}\in\Kset_{i,t}\),
\[
\sum_{j=1}^s p_{ij,t}f(j)\ge \min_{1\le j\le s}f(j),
\]
which gives the lower bound after taking the infimum. For superadditivity,
\[
\inf_{p\in\Kset_{i,t}}p(f+h)
\ge
\inf_{p\in\Kset_{i,t}}pf+\inf_{p\in\Kset_{i,t}}ph,
\]
where \(pf\) denotes \(\sum_{j=1}^s p_jf(j)\). Non-negative homogeneity follows because, for \(\mu\ge0\),
\[
\inf_{p\in\Kset_{i,t}}p(\mu f)=\mu\inf_{p\in\Kset_{i,t}}pf.
\]
Constant additivity follows from \(p\mathbf 1=1\) for every probability row \(p\). Monotonicity follows because \(f\le h\) implies \(pf\le ph\) for every admissible row. Finally,
\[
-\lowerT_t(-f)(i)
=
-\inf_{p\in\Kset_{i,t}}\sum_{j=1}^s p_{ij,t}\{-f(j)\}
=
\sup_{p\in\Kset_{i,t}}\sum_{j=1}^s p_{ij,t}f(j)
=
\upperT_t f(i).
\]

\subsection{Proof of Theorem~\ref{thm:bellman-envelope}}

We prove the lower recursion; the upper recursion is identical with infima replaced by suprema. For an admissible sequence \(P_{0:T-1}\in\Aset\), let \(V_t^P\) denote the classical backward vector \eqref{eq:classical-bellman}, so that \(m_0V_0^P=A(P_{0:T-1})\) by the forward--backward equivalence. We establish two facts: (i) \(\lowerV_t\le V_t^P\) componentwise for every admissible \(P\) and every \(t\); and (ii) there exists an admissible \(P^\star\) with \(V_t^{P^\star}=\lowerV_t\). Both use only Assumption~\ref{ass:compact-row-sets}.

\emph{Lower bound.} Backward induction on \(t\). At \(t=T\), \(\lowerV_T=g_T=V_T^P\). Assume \(\lowerV_{t+1}\le V_{t+1}^P\). Each row \(p_{i,t}\) of \(P_t\) has non-negative entries, so for every state \(S_i\),
\[
\begin{aligned}
V_t^P(i)=g_t(i)+\sum_{j=1}^s p_{ij,t}V_{t+1}^P(j)
&\ge g_t(i)+\sum_{j=1}^s p_{ij,t}\lowerV_{t+1}(j)\\
&\ge g_t(i)+\inf_{p_{i,t}\in\Kset_{i,t}}\sum_{j=1}^s p_{ij,t}\lowerV_{t+1}(j)
=\lowerV_t(i),
\end{aligned}
\]
where the first inequality uses the induction hypothesis with \(p_{ij,t}\ge0\). Hence \(\lowerV_t\le V_t^P\) for all \(t\). Since \(m_0\) has non-negative entries, \(A(P_{0:T-1})=m_0V_0^P\ge m_0\lowerV_0\) for every admissible \(P\), and therefore \(m_0\lowerV_0\le\inf_{P_{0:T-1}\in\Aset}A(P_{0:T-1})\).

\emph{Attainment.} By Assumption~\ref{ass:compact-row-sets} the map \(p_{i,t}\mapsto\sum_{j=1}^s p_{ij,t}\lowerV_{t+1}(j)\) is linear, hence continuous, on the compact set \(\Kset_{i,t}\), so a minimizer \(p^\star_{i,t}\in\Kset_{i,t}\) exists. Because the rows are separately specified \eqref{eq:row-separable-set}, stacking \(p^\star_{1,t},\ldots,p^\star_{s,t}\) yields a matrix \(P^\star_t\in\Pset_t\); repeating this for every time step gives \(P^\star_{0:T-1}\in\Aset\). A second backward induction shows \(V_t^{P^\star}=\lowerV_t\): it holds at \(t=T\), and if \(V_{t+1}^{P^\star}=\lowerV_{t+1}\) then
\[
V_t^{P^\star}(i)=g_t(i)+\sum_{j=1}^s p^\star_{ij,t}\lowerV_{t+1}(j)=g_t(i)+(\lowerT_t\lowerV_{t+1})(i)=\lowerV_t(i).
\]
Therefore \(A(P^\star_{0:T-1})=m_0V_0^{P^\star}=m_0\lowerV_0\), the infimum is attained, and combined with the lower bound, \(m_0\lowerV_0=\inf_{P_{0:T-1}\in\Aset}A(P_{0:T-1})\). The minimizer \(P^\star\) is constructed without reference to \(m_0\), so the same sequence attains the lower value for every initial distribution.

\subsection{Proof of Corollary~\ref{cor:singleton}}

If \(\Pset_t=\{P_t\}\), each \(\Kset_{i,t}\) contains the single row \(P_{t,i\cdot}\). Therefore both the infimum and supremum in \eqref{eq:lower-upper-operators-ijar} equal \(\sum_{j=1}^s p_{ij,t}f(j)\), which is \((P_tf)(i)\). Substituting into \eqref{eq:imprecise-bellman-ijar} gives \eqref{eq:classical-bellman}.

\subsection{Proof of Proposition~\ref{prop:performance-difference}}

Let \(V_t^{\widehat P}\) be the backward vector under \(\widehat P\). For any admissible sequence \(Q\),
\[
A(Q)-A(\widehat P)=m_0(V_0^Q-V_0^{\widehat P}).
\]
Using the backward recursions,
\[
V_t^Q-V_t^{\widehat P}
=
Q_t(V_{t+1}^Q-V_{t+1}^{\widehat P})+(Q_t-\widehat P_t)V_{t+1}^{\widehat P}.
\]
Multiplying by \(m_t^Q\) and using \(m_{t+1}^Q=m_t^QQ_t\) yields
\[
m_t^Q(V_t^Q-V_t^{\widehat P})
=
m_{t+1}^Q(V_{t+1}^Q-V_{t+1}^{\widehat P})
+
m_t^Q(Q_t-\widehat P_t)V_{t+1}^{\widehat P}.
\]
Summing over \(t=0,\ldots,T-1\) telescopes. The terminal difference is zero because \(V_T^Q=V_T^{\widehat P}=g_T\). This gives \eqref{eq:performance-difference}.

\subsection{Proof of Proposition~\ref{prop:decision-robustness}}

Under Assumption~\ref{ass:compact-row-sets} the joint admissible set \(\Aset_{ab}\) is compact and the map \((P^a,P^b)\mapsto D(P^a,P^b;\lambda)\) is continuous, so by the extreme value theorem the infimum \(\underline D\) is attained at some admissible pair. If \(\underline D>0\), then \(D(P^a,P^b;\lambda)\ge\underline D>0\) for every admissible pair, so intervention \(a\) has strictly greater net monetary benefit throughout \(\Aset_{ab}\) and the recommendation is robust. If \(\underline D\le0\), a minimizing pair satisfies \(D\le0\), so the strict guarantee fails. If moreover \(\underline D<0\), that minimizing pair has \(D(P^a,P^b;\lambda)<0\) and hence favors intervention \(b\).

\bibliographystyle{elsarticle-num}
\bibliography{references}

@book{Augustin2014,
  title     = {Introduction to Imprecise Probabilities},
  editor    = {Thomas Augustin and Frank P. A. Coolen and Gert de Cooman and Matthias C. M. Troffaes},
  year      = {2014},
  publisher = {John Wiley \& Sons},
  address   = {Chichester, UK}
}

@book{Walley1991,
  title     = {Statistical Reasoning with Imprecise Probabilities},
  author    = {Walley, Peter},
  year      = {1991},
  publisher = {Chapman and Hall},
  address   = {London}
}

@article{BeckPauker1983,
  title   = {The {M}arkov Process in Medical Prognosis},
  author  = {Beck, J. Robert and Pauker, Stephen G.},
  journal = {Medical Decision Making},
  year    = {1983},
  volume  = {3},
  number  = {4},
  pages   = {419--458}
}

@article{SonnenbergBeck1993,
  title   = {{M}arkov Models in Medical Decision Making: A Practical Guide},
  author  = {Sonnenberg, Frank A. and Beck, J. Robert},
  journal = {Medical Decision Making},
  year    = {1993},
  volume  = {13},
  number  = {4},
  pages   = {322--338}
}

@article{Iskandar2018,
  title   = {A Theoretical Foundation for State-Transition Cohort Models in Health Decision Analysis},
  author  = {Iskandar, Rowan},
  journal = {PLOS ONE},
  year    = {2018},
  volume  = {13},
  number  = {12},
  pages   = {e0205543}
}

@incollection{HermansSkulj2014,
  title     = {Stochastic Processes},
  author    = {Hermans, Filip and {\v S}kulj, Damjan},
  booktitle = {Introduction to Imprecise Probabilities},
  editor    = {Augustin, Thomas and Coolen, Frank P. A. and de Cooman, Gert and Troffaes, Matthias C. M.},
  year      = {2014},
  publisher = {John Wiley \& Sons},
  pages     = {258--278}
}

@article{KrakDeBockSiebes2017,
  title   = {Imprecise Continuous-Time {M}arkov Chains},
  author  = {Krak, Thomas and De Bock, Jasper and Siebes, Arno},
  journal = {International Journal of Approximate Reasoning},
  year    = {2017},
  volume  = {88},
  pages   = {452--528}
}

@article{Mas2017CLOSE,
  title   = {Patent Foramen Ovale Closure or Anticoagulation vs. Antiplatelets after Stroke},
  author  = {Mas, Jean-Louis and Derumeaux, Genevi{\`e}ve and Guillon, Benoit and Massardier, Emmanuel and Hosseini, Hassan and Mechtouff, Laura and Arquizan, Caroline and B{\'e}jot, Yannick and Vuillier, Fran{\c c}ois and Detante, Olivier and Guidoux, C{\'e}cile and Canaple, Sophie and Vaduva, Cosmin and Dequatre-Ponchelle, Nadine and Sibon, Igor and Garnier, Philippe and Ferrier, Arnaud and Timsit, Serge and Robinet-Borgomano, Emmanuelle and Sablot, Denis and Lacour, Jean-Charles and Zuber, Marc and Favrole, Patricia and Pinel, Jean-Fran{\c c}ois and Apoil, Matthieu and Reiner, Philippe and Leftheriotis, Georges and Gu{\'e}rin, Philippe and Piot, Christophe and Rossi, Renaud and Dubois-Rand{\'e}, Jean-Luc and Eicher, Jean-Claude and Meneveau, Nicolas and Lusson, Jean-Ren{\'e} and Bertrand, Bruno and Schleich, Jean-Marc and Godart, Fran{\c c}ois and Thambo, Jean-Beno{\^i}t and Leborgne, Lawrence and Michel, Patrice and Pierard, Luc and Turc, Guillaume and Barthelet, Michel and Charles-Nelson, Aur{\'e}lie and Weimar, Christian and Moulin, Thierry and Juliard, Jean-Michel and Chatellier, Gilles},
  journal = {New England Journal of Medicine},
  year    = {2017},
  volume  = {377},
  pages   = {1011--1021}
}

@article{Lee2018DEFENSEPFO,
  title   = {Cryptogenic Stroke and High-Risk Patent Foramen Ovale: The {DEFENSE-PFO} Trial},
  author  = {Lee, Pil Hyung and Song, Jae-Kwan and Kim, Jong S. and Heo, Ji Hoe and Lee, Seung-Woon and Kim, Dong-Hyun and Song, Jae-Hwan and Kang, Dong-Wha and Kwon, Sang-Wook and Kang, Dae-Hyeok and Lee, Cheol Whan and Hong, Myeong-Ki and Kim, Jae-Joong and Park, Seung-Jung},
  journal = {Journal of the American College of Cardiology},
  year    = {2018},
  volume  = {71},
  pages   = {2335--2342}
}

@article{Leppert2018PFOCEA,
  title   = {Cost-Effectiveness of Patent Foramen Ovale Closure Versus Medical Therapy for Secondary Stroke Prevention},
  author  = {Leppert, Michelle H. and Poisson, Sharon N. and Carroll, John D. and Thaler, David E. and Kim, Catherine H. and Orjuela, Karen D. and Wiggins, Benjamin S. and Berman, Matt and Velez, Carlos A.},
  journal = {Stroke},
  year    = {2018},
  volume  = {49},
  pages   = {1443--1450}
}

@article{Skulj2015,
  title   = {Efficient Computation of the Bounds of Continuous Time Imprecise {M}arkov Chains},
  author  = {{\v S}kulj, Damjan},
  journal = {Applied Mathematics and Computation},
  year    = {2015},
  volume  = {250},
  pages   = {165--180}
}

@article{DeCoomanHermansQuaeghebeur2009,
  title   = {Imprecise {M}arkov Chains and Their Limit Behavior},
  author  = {de Cooman, Gert and Hermans, Filip and Quaeghebeur, Erik},
  journal = {Probability in the Engineering and Informational Sciences},
  year    = {2009},
  volume  = {23},
  number  = {4},
  pages   = {597--635}
}

@book{TroffaesDeCooman2014,
  title     = {Lower Previsions},
  author    = {Troffaes, Matthias C. M. and de Cooman, Gert},
  year      = {2014},
  publisher = {Wiley},
  address   = {Chichester}
}

@article{DeCoomanDeBockLopatatzidis2016,
  title   = {Imprecise Stochastic Processes in Discrete Time: Global Models, Imprecise {M}arkov Chains, and Ergodic Theorems},
  author  = {de Cooman, Gert and De Bock, Jasper and Lopatatzidis, Stavros},
  journal = {International Journal of Approximate Reasoning},
  year    = {2016},
  volume  = {76},
  pages   = {18--46}
}

@article{TJoensDeBock2021,
  title   = {Average Behaviour in Discrete-Time Imprecise {M}arkov Chains: A Study of Weak Ergodicity},
  author  = {T'Joens, Natan and De Bock, Jasper},
  journal = {International Journal of Approximate Reasoning},
  year    = {2021},
  volume  = {132},
  pages   = {181--205}
}

@article{Iyengar2005,
  title   = {Robust Dynamic Programming},
  author  = {Iyengar, Garud N.},
  journal = {Mathematics of Operations Research},
  year    = {2005},
  volume  = {30},
  number  = {2},
  pages   = {257--280}
}

@article{NilimElGhaoui2005,
  title   = {Robust Control of {M}arkov Decision Processes with Uncertain Transition Matrices},
  author  = {Nilim, Arnab and El Ghaoui, Laurent},
  journal = {Operations Research},
  year    = {2005},
  volume  = {53},
  number  = {5},
  pages   = {780--798}
}

@article{Wiesemann2013,
  title   = {Robust {M}arkov Decision Processes},
  author  = {Wiesemann, Wolfram and Kuhn, Daniel and Rustem, Ber{\c c}},
  journal = {Mathematics of Operations Research},
  year    = {2013},
  volume  = {38},
  number  = {1},
  pages   = {153--183}
}

\end{document}


\maketitle
\tableofcontents
\bigskip
This document consolidates the supplementary materials for the
manuscript. \textbf{Part I} is a step-by-step implementation guide for
the patent foramen ovale (PFO) case study. \textbf{Part II} lists the
full source of the accompanying R scripts. The individual runnable
source files (the \passthrough{\lstinline!.R!} scripts and the original
Quarto guide \passthrough{\lstinline!implementation\_guide.qmd!}) are
also provided in the submission package so that all results can be
regenerated.

\hypertarget{part-i.-implementation-guide}{%
\section{Part I. Implementation
guide}\label{part-i.-implementation-guide}}

\hypertarget{overview}{%
\subsection{Overview}\label{overview}}

This Quarto file implements the real-data case study from the
manuscript.

The clinical example is \textbf{patent foramen ovale (PFO) closure after
cryptogenic stroke}. The goal is to show how a classical Markov cohort
model changes when the transition probability matrix is not treated as
known exactly.

The example compares three analyses:

\begin{enumerate}
\def\labelenumi{\arabic{enumi}.}
\tightlist
\item
  \textbf{Empirical transition matrix}: convert observed
  recurrent-stroke counts directly into annual transition probabilities.
\item
  \textbf{Precise Dirichlet posterior mean}: use one precise posterior
  mean for the recurrent-stroke probability.
\item
  \textbf{Imprecise Dirichlet Model (IDM)}: use a set of admissible
  recurrent-stroke probabilities rather than one probability.
\end{enumerate}

The key decision question is whether PFO closure is robustly favored at
a willingness-to-pay threshold of \$150,000 per quality-adjusted
life-year (QALY).

\hypertarget{trial-evidence-for-recurrent-stroke}{%
\subsection{1. Trial evidence for recurrent
stroke}\label{trial-evidence-for-recurrent-stroke}}

The model uses two sources of high-risk PFO evidence:

\begin{itemize}
\tightlist
\item
  CLOSE: 0 recurrent strokes among 238 patients assigned to PFO closure,
  and 14 recurrent strokes among 235 patients assigned to antiplatelet
  therapy.
\item
  DEFENSE-PFO: 0 recurrent strokes among 60 patients assigned to PFO
  closure, and 10.5\% ischemic stroke in the medical-therapy arm. With
  60 patients, this is represented as 6 events.
\end{itemize}

For this simple reproducible example, patient-years are approximated as
patients multiplied by follow-up duration.

\begin{lstlisting}
trials <- data.frame(
  trial = c("CLOSE", "CLOSE", "DEFENSE-PFO", "DEFENSE-PFO"),
  strategy = c("PFO closure", "Medical therapy", "PFO closure", "Medical therapy"),
  patients = c(238, 235, 60, 60),
  followup_time_steps = c(5.3, 5.3, 2.8, 2.8),
  recurrent_strokes = c(0, 14, 0, 6)
)

trials$patient_years <- trials$patients * trials$followup_time_steps
trials
\end{lstlisting}

Pool the transition counts by strategy.

\begin{lstlisting}
closure_events <- sum(trials$recurrent_strokes[trials$strategy == "PFO closure"])
medical_events <- sum(trials$recurrent_strokes[trials$strategy == "Medical therapy"])

closure_py <- sum(trials$patient_years[trials$strategy == "PFO closure"])
medical_py <- sum(trials$patient_years[trials$strategy == "Medical therapy"])

evidence <- data.frame(
  strategy = c("PFO closure", "Medical therapy"),
  events = c(closure_events, medical_events),
  patient_years = c(closure_py, medical_py)
)

evidence
\end{lstlisting}

\hypertarget{empirical-and-idm-transition-probabilities}{%
\subsection{2. Empirical and IDM transition
probabilities}\label{empirical-and-idm-transition-probabilities}}

The empirical annual recurrent-stroke probability is

\[
\widehat p = \frac{\text{events}}{\text{patient-years}}.
\]

For the IDM, use strength \(\alpha=2\). For a binary row, the lower and
upper bounds for the recurrent-stroke probability are

\[
\underline p = \frac{x}{N+\alpha},
\qquad
\overline p = \frac{x+\alpha}{N+\alpha},
\]

where \(x\) is the recurrent-stroke count and \(N\) is the approximate
patient-years at risk.

A precise symmetric Dirichlet posterior mean is also shown as a
contrast:

\[
p^{\mathrm{Dir}} = \frac{x+\alpha/2}{N+\alpha}.
\]

\begin{lstlisting}
alpha <- 2

compute_probabilities <- function(events, exposure, alpha = 2) {
  c(
    empirical = events / exposure,
    dirichlet_mean = (events + alpha / 2) / (exposure + alpha),
    idm_lower = events / (exposure + alpha),
    idm_upper = (events + alpha) / (exposure + alpha)
  )
}

prob_closure <- compute_probabilities(closure_events, closure_py, alpha)
prob_medical <- compute_probabilities(medical_events, medical_py, alpha)

transition_probabilities <- rbind(
  "PFO closure" = prob_closure,
  "Medical therapy" = prob_medical
)

round(transition_probabilities, 6)
\end{lstlisting}

The important point is that the empirical closure-arm probability is
zero, but the IDM upper probability is not zero. The observed zero is
evidence, not proof that recurrent stroke is impossible.

\hypertarget{states-and-model-inputs}{%
\subsection{3. States and model inputs}\label{states-and-model-inputs}}

Use three states:

\begin{enumerate}
\def\labelenumi{\arabic{enumi}.}
\tightlist
\item
  No recurrent stroke.
\item
  After recurrent stroke.
\item
  Dead.
\end{enumerate}

Everyone begins in the no-recurrent-stroke state.

\begin{lstlisting}
states <- c("No recurrent stroke", "After recurrent stroke", "Dead")

m0 <- c("No recurrent stroke" = 1,
        "After recurrent stroke" = 0,
        "Dead" = 0)

horizon <- 7
discount_rate <- 0.03
lambda <- 150000

# Background mortality assumptions
p_death_no_stroke <- 0.002
p_death_post_stroke <- 0.050

# State values
utility_no_stroke <- 0.86
utility_post_stroke <- 0.65

# Costs
cost_closure_once <- 15000
cost_no_stroke <- 200
cost_acute_stroke <- 40000
cost_post_stroke <- 8000
\end{lstlisting}

\hypertarget{build-a-transition-probability-matrix}{%
\subsection{4. Build a transition probability
matrix}\label{build-a-transition-probability-matrix}}

For a given recurrent-stroke probability \(p\), the transition matrix is

\[
P =
\begin{pmatrix}
1-p-d_0 & p & d_0\\
0 & 1-d_1 & d_1\\
0 & 0 & 1
\end{pmatrix},
\]

where \(d_0\) is the annual background death probability before
recurrent stroke and \(d_1\) is the death probability after recurrent
stroke.

\begin{lstlisting}
make_transition_matrix <- function(p_stroke,
                                   p_death_no_stroke = 0.002,
                                   p_death_post_stroke = 0.050) {
  P <- matrix(
    c(
      1 - p_stroke - p_death_no_stroke, p_stroke, p_death_no_stroke,
      0, 1 - p_death_post_stroke, p_death_post_stroke,
      0, 0, 1
    ),
    nrow = 3,
    byrow = TRUE,
    dimnames = list(states, states)
  )

  # Basic checks
  stopifnot(all(P >= -1e-12))
  stopifnot(max(abs(rowSums(P) - 1)) < 1e-10)

  P
}

make_transition_matrix(prob_medical["empirical"])
\end{lstlisting}

\hypertarget{run-the-markov-cohort-model}{%
\subsection{5. Run the Markov cohort
model}\label{run-the-markov-cohort-model}}

The cohort state vector is updated as

\[
m_{t+1}=m_tP.
\]

Time-step outcome contributions are calculated from the current cohort
distribution. We discount both costs and QALYs at 3\% in this compact
example.

\begin{lstlisting}
evaluate_strategy <- function(p_stroke,
                              closure = FALSE,
                              horizon = 7,
                              lambda = 150000,
                              discount_rate = 0.03) {
  P <- make_transition_matrix(
    p_stroke,
    p_death_no_stroke = p_death_no_stroke,
    p_death_post_stroke = p_death_post_stroke
  )

  m <- m0
  total_qalys <- 0
  total_costs <- if (closure) cost_closure_once else 0

  trajectory <- data.frame(
    time_step = 0,
    no_recurrent_stroke = m[1],
    after_recurrent_stroke = m[2],
    dead = m[3],
    time_step_qalys = NA_real_,
    time_step_costs = NA_real_
  )

  for (t in 0:(horizon - 1)) {
    discount <- 1 / (1 + discount_rate)^t

    time_step_qalys <- (
      m["No recurrent stroke"] * utility_no_stroke +
      m["After recurrent stroke"] * utility_post_stroke
    ) * discount

    time_step_costs <- (
      m["No recurrent stroke"] * cost_no_stroke +
      m["After recurrent stroke"] * cost_post_stroke
    ) * discount

    # Acute stroke costs are counted for the fraction of the cohort that
    # transitions from no recurrent stroke to recurrent stroke during the time step.
    new_strokes <- m["No recurrent stroke"] * p_stroke
    acute_costs <- new_strokes * cost_acute_stroke * discount

    total_qalys <- total_qalys + time_step_qalys
    total_costs <- total_costs + time_step_costs + acute_costs

    m <- as.numeric(m %*% P)
    names(m) <- states

    trajectory <- rbind(
      trajectory,
      data.frame(
        time_step = t + 1,
        no_recurrent_stroke = m[1],
        after_recurrent_stroke = m[2],
        dead = m[3],
        time_step_qalys = time_step_qalys,
        time_step_costs = time_step_costs + acute_costs
      )
    )
  }

  list(
    qalys = total_qalys,
    costs = total_costs,
    nmb = lambda * total_qalys - total_costs,
    trajectory = trajectory
  )
}
\end{lstlisting}

\hypertarget{compare-strategies}{%
\subsection{6. Compare strategies}\label{compare-strategies}}

The incremental net monetary benefit (INMB) for PFO closure versus
medical therapy is

\[
\mathrm{INMB}
=
\lambda(Q_{\mathrm{closure}}-Q_{\mathrm{medical}})
-
(C_{\mathrm{closure}}-C_{\mathrm{medical}}).
\]

\begin{lstlisting}
compare_strategies <- function(p_closure, p_medical, label) {
  closure <- evaluate_strategy(p_closure, closure = TRUE)
  medical <- evaluate_strategy(p_medical, closure = FALSE)

  inc_qalys <- closure$qalys - medical$qalys
  inc_costs <- closure$costs - medical$costs
  inmb <- closure$nmb - medical$nmb

  data.frame(
    analysis = label,
    incremental_qalys = inc_qalys,
    incremental_costs = inc_costs,
    INMB = inmb
  )
}
\end{lstlisting}

\hypertarget{empirical-transition-matrix}{%
\subsubsection{Empirical transition
matrix}\label{empirical-transition-matrix}}

\begin{lstlisting}
empirical_result <- compare_strategies(
  p_closure = prob_closure["empirical"],
  p_medical = prob_medical["empirical"],
  label = "Empirical transition matrix"
)

empirical_result
\end{lstlisting}

\hypertarget{precise-dirichlet-posterior-mean}{%
\subsubsection{Precise Dirichlet posterior
mean}\label{precise-dirichlet-posterior-mean}}

\begin{lstlisting}
dirichlet_result <- compare_strategies(
  p_closure = prob_closure["dirichlet_mean"],
  p_medical = prob_medical["dirichlet_mean"],
  label = "Precise Dirichlet posterior mean"
)

dirichlet_result
\end{lstlisting}

\hypertarget{imprecise-analysis}{%
\subsection{7. Imprecise analysis}\label{imprecise-analysis}}

To get the \textbf{lower} INMB for PFO closure, use the least favorable
admissible comparison:

\begin{itemize}
\tightlist
\item
  closure recurrent-stroke probability at its IDM upper bound;
\item
  medical-therapy recurrent-stroke probability at its IDM lower bound.
\end{itemize}

To get the \textbf{upper} INMB, reverse those choices:

\begin{itemize}
\tightlist
\item
  closure recurrent-stroke probability at its IDM lower bound;
\item
  medical-therapy recurrent-stroke probability at its IDM upper bound.
\end{itemize}

\begin{lstlisting}
idm_lower <- compare_strategies(
  p_closure = prob_closure["idm_upper"],
  p_medical = prob_medical["idm_lower"],
  label = "IDM lower bound"
)

idm_upper <- compare_strategies(
  p_closure = prob_closure["idm_lower"],
  p_medical = prob_medical["idm_upper"],
  label = "IDM upper bound"
)

results <- rbind(
  empirical_result,
  dirichlet_result,
  idm_lower,
  idm_upper
)

results$incremental_costs <- round(results$incremental_costs, 0)
results$incremental_qalys <- round(results$incremental_qalys, 4)
results$INMB <- round(results$INMB, 0)

results
\end{lstlisting}

\hypertarget{the-general-method-lower-and-upper-transition-operator-recursion}{%
\subsection{8. The general method: lower and upper transition-operator
recursion}\label{the-general-method-lower-and-upper-transition-operator-recursion}}

Section 7 obtained the IDM interval by evaluating the classical cohort
model at the interval endpoints. That shortcut is valid here only
because incremental net monetary benefit is monotone in the single
recurrent-stroke probability, so the worst and best admissible rows are
the same at every time step (Remark 1 in the manuscript). The general
method in Theorem 1 does not rely on that. It applies the \textbf{lower
and upper transition operators} at every step,

\[
(\underline{\mathcal T}_t V)(i)=\inf_{p_i\in\mathcal K_{i,t}}\sum_j p_{ij}\bigl(r_t(i,j)+V(j)\bigr),
\qquad
(\overline{\mathcal T}_t V)(i)=\sup_{p_i\in\mathcal K_{i,t}}\sum_j p_{ij}\bigl(r_t(i,j)+V(j)\bigr),
\]

and runs the backward recursion
\(V_t=g_t+\underline{\mathcal T}_t V_{t+1}\) for the lower envelope (or
\(\overline{\mathcal T}_t\) for the upper). The transition reward
\(r_t(i,j)\) carries outcomes incurred on a transition, such as the
acute recurrent-stroke cost. We implement this directly and confirm it
reproduces the Section 7 interval.

\hypertarget{one-step-transition-operator-over-a-row-set}{%
\subsubsection{8.1 One-step transition operator over a row
set}\label{one-step-transition-operator-over-a-row-set}}

For a row set
\(\mathcal K=\{p:\ \underline p\le p\le\overline p,\ \sum_j p_j=1\}\)
and a target vector \(\mathrm{target}_j=r(i,j)+V(j)\), the lower (upper)
one-step value is the minimum (maximum) of a linear function over a
box-truncated simplex. It is found greedily: start at the lower bounds
and allocate the remaining mass \(1-\sum_j\underline p_j\) to the
smallest-target coordinates for the lower operator, or the largest for
the upper.

\begin{lstlisting}
row_operator <- function(target, p_lo, p_hi, sense = c("lower", "upper")) {
  sense <- match.arg(sense)
  stopifnot(sum(p_lo) <= 1 + 1e-9, sum(p_hi) >= 1 - 1e-9)
  free <- 1 - sum(p_lo)
  ord <- order(target, decreasing = (sense == "upper"))
  p <- p_lo
  for (j in ord) {
    take <- min(free, p_hi[j] - p_lo[j])
    p[j] <- p[j] + take
    free <- free - take
  }
  sum(p * target)
}
\end{lstlisting}

For a full Imprecise Dirichlet Model row this has the closed form
\(\underline{\mathcal T}f=\bigl(\sum_j n_{ij}f(j)+\alpha\min_j f(j)\bigr)/(N_i+\alpha)\),
and the maximum for the upper operator; the greedy version above also
handles partially constrained rows, such as the one-parameter
recurrent-stroke row used here.

\hypertarget{backward-recursion-in-transition-reward-form}{%
\subsubsection{8.2 Backward recursion in transition-reward
form}\label{backward-recursion-in-transition-reward-form}}

\begin{lstlisting}
arm_envelope <- function(p_lo, p_hi, sense, closure) {
  # State order: 1 = no recurrent stroke, 2 = after recurrent stroke, 3 = dead.
  # Only the no-stroke row is imprecise (stroke probability in [p_lo, p_hi]);
  # background death is held fixed. The acute stroke cost is a transition reward.
  p_lo <- unname(p_lo); p_hi <- unname(p_hi)
  V <- c(0, 0, 0)                                   # terminal value V_T = 0
  for (t in (horizon - 1):0) {
    discount <- 1 / (1 + discount_rate)^t
    g <- discount * (lambda * c(utility_no_stroke, utility_post_stroke, 0) -
                     c(cost_no_stroke, cost_post_stroke, 0))
    rows_lo <- list(c(1 - p_hi - p_death_no_stroke, p_lo, p_death_no_stroke),
                    c(0, 1 - p_death_post_stroke, p_death_post_stroke),
                    c(0, 0, 1))
    rows_hi <- list(c(1 - p_lo - p_death_no_stroke, p_hi, p_death_no_stroke),
                    c(0, 1 - p_death_post_stroke, p_death_post_stroke),
                    c(0, 0, 1))
    reward <- matrix(0, 3, 3)
    reward[1, 2] <- -discount * cost_acute_stroke   # acute cost: no-stroke -> post-stroke
    Vnew <- numeric(3)
    for (i in 1:3) {
      target <- reward[i, ] + V                     # r_t(i,.) + V_{t+1}
      Vnew[i] <- g[i] + row_operator(target, rows_lo[[i]], rows_hi[[i]], sense)
    }
    V <- Vnew
  }
  V[1] - if (closure) cost_closure_once else 0       # m0 = e_1; subtract one-time closure cost
}
\end{lstlisting}

\hypertarget{lower-and-upper-inmb-from-the-operators}{%
\subsubsection{8.3 Lower and upper INMB from the
operators}\label{lower-and-upper-inmb-from-the-operators}}

Because the two arms are informed by separate evidence, the admissible
set is a product and the INMB endpoints factorize: the lower INMB pairs
the worst closure arm with the best medical arm, and the upper INMB does
the reverse.

\begin{lstlisting}
pc_lo <- prob_closure["idm_lower"]; pc_hi <- prob_closure["idm_upper"]
pm_lo <- prob_medical["idm_lower"]; pm_hi <- prob_medical["idm_upper"]

lower_INMB_operator <- arm_envelope(pc_lo, pc_hi, "lower", TRUE) -
                       arm_envelope(pm_lo, pm_hi, "upper", FALSE)
upper_INMB_operator <- arm_envelope(pc_lo, pc_hi, "upper", TRUE) -
                       arm_envelope(pm_lo, pm_hi, "lower", FALSE)

data.frame(
  method = c("Endpoint evaluation (Section 7)",
             "Transition-operator recursion (Theorem 1)"),
  lower_INMB = round(c(idm_lower$INMB, lower_INMB_operator)),
  upper_INMB = round(c(idm_upper$INMB, upper_INMB_operator))
)
\end{lstlisting}

The two methods agree. The endpoint evaluation is the special case in
which the optimizing row is the same at every time step; the
transition-operator recursion is the method to use when that is not
guaranteed---when outcome contributions or transition rewards vary
across time in a way that flips the optimal row, when more than one row
of the matrix is imprecise, or when the outcome is not monotone in the
uncertain parameters. To adapt it to another model, replace the state
rewards \passthrough{\lstinline!g!}, the transition rewards
\passthrough{\lstinline!reward!}, and the row-bound lists
\passthrough{\lstinline!rows\_lo!}/\passthrough{\lstinline!rows\_hi!};
the operator and the recursion are unchanged.

\hypertarget{decision-interpretation}{%
\subsection{9. Decision interpretation}\label{decision-interpretation}}

The empirical transition matrix gives a slightly positive INMB for PFO
closure. A classical point analysis would therefore favor closure at
this threshold.

The IDM interval crosses zero, so the imprecise analysis does
\textbf{not} support a robust recommendation at this threshold.

\begin{lstlisting}
emp_inmb <- empirical_result$INMB
lower_inmb <- idm_lower$INMB
upper_inmb <- idm_upper$INMB

cat("Empirical INMB:", round(emp_inmb, 0), "\n")
cat("IDM INMB interval: [", round(lower_inmb, 0), ", ", round(upper_inmb, 0), "]\n", sep = "")

if (emp_inmb > 0 && lower_inmb < 0 && upper_inmb > 0) {
  cat("Conclusion: the point model favors PFO closure, but the imprecise model shows non-robustness.\n")
}
\end{lstlisting}

A small plot helps visualize the result.

\begin{lstlisting}
plot(
  x = c(lower_inmb, upper_inmb),
  y = c(1, 1),
  type = "l",
  lwd = 5,
  xlab = "Incremental net monetary benefit for PFO closure",
  ylab = "",
  yaxt = "n",
  ylim = c(0.7, 1.3),
  main = "Empirical point estimate versus IDM interval"
)
axis(2, at = 1, labels = "IDM interval", las = 1)
abline(v = 0, lty = 2)
points(emp_inmb, 1, pch = 19)
text(emp_inmb, 1.1, labels = paste0("empirical = $", round(emp_inmb, 0)))
text(lower_inmb, 0.9, labels = paste0("$", round(lower_inmb, 0)), pos = 2)
text(upper_inmb, 0.9, labels = paste0("$", round(upper_inmb, 0)), pos = 4)
\end{lstlisting}

\hypertarget{what-this-example-shows}{%
\subsection{10. What this example shows}\label{what-this-example-shows}}

The empirical model treats the observed zero recurrent-stroke count in
the PFO closure arms as a zero recurrent-stroke probability. The IDM
model treats the observed zero as evidence for a small probability,
while still allowing non-zero values compatible with the trial evidence.

That difference changes the decision interpretation near the threshold:

\begin{itemize}
\tightlist
\item
  The empirical transition matrix says: \textbf{favor PFO closure}.
\item
  The imprecise transition model says: \textbf{the recommendation is not
  robust}.
\end{itemize}

This is the main practical reason to use an imprecise Markov cohort
model when transition evidence is sparse or only partially resolves the
transition probability matrix.

\hypertarget{part-ii.-r-source-scripts}{%
\section{Part II. R source scripts}\label{part-ii.-r-source-scripts}}

The following scripts reproduce every numerical result in the
manuscript. Each is self-contained and can be run with a base R
installation.

\hypertarget{bellman_envelope.r}{%
\subsection{\texorpdfstring{\texttt{bellman\_envelope.R}}{bellman\_envelope.R}}\label{bellman_envelope.r}}

Independent verification of the case-study interval via the lower and
upper transition-operator recursion (Theorem 1, transition-reward form).

\begin{lstlisting}
# bellman_envelope.R
# Direct implementation of Theorem 1 (lower/upper transition-operator Bellman
# recursion) for the PFO imprecise Markov cohort model of the case study.
#
# Unlike case_study_reproducible.R, which evaluates the classical cohort model at
# the IDM interval endpoints, this script computes the lower and upper envelopes
# from the transition operators in Eq. (imprecise-bellman-transition), folding in
# the acute recurrent-stroke cost as a one-step transition reward r(1,2).
# It recovers the same INMB interval, illustrating Remark 1: because INMB is
# monotone in the single recurrent-stroke probability, the optimizing row is
# pinned to an interval endpoint at every step, so the rectangular envelope and
# the fixed-matrix endpoint evaluation coincide here.
#
# States: 1 = no recurrent stroke, 2 = after recurrent stroke, 3 = dead.

lambda         <- 150000
discount_rate  <- 0.03
horizon        <- 7
u              <- c(0.86, 0.65, 0.00)   # utilities by state
c_state        <- c(200, 8000, 0.00)    # annual state-occupancy costs
d0             <- 0.002                 # death probability, no-stroke state
d1             <- 0.050                 # death probability, post-stroke state
c_acute        <- 40000                 # acute transition cost: no-stroke -> post-stroke
c_closure_once <- 15000                 # one-time closure cost

# One backward step of the lower/upper transition-operator recursion for this model.
# Only the no-stroke row (state 1) is imprecise: p_stroke in [p_lo, p_hi], with the
# background death probability d0 held fixed (the one-parameter subfamily of the IDM
# simplex used in the case study). Rows 2 and 3 are precise. The acute cost enters as the
# transition reward r(1,2) = -dt * c_acute in net-monetary-benefit units.
step_value <- function(Vnext, p_lo, p_hi, dt, sense) {
  Vt <- numeric(3)
  Vt[2] <- (1 - d1) * Vnext[2] + d1 * Vnext[3]   # post-stroke row [0, 1-d1, d1]
  Vt[3] <- Vnext[3]                              # dead row (absorbing)
  # no-stroke row value as a function of p:
  #   (1 - d0) Vnext[1] + d0 Vnext[3] + p * ( Vnext[2] - Vnext[1] - dt * c_acute )
  coef <- Vnext[2] - Vnext[1] - dt * c_acute
  p <- if (sense == "lower") {
    if (coef < 0) p_hi else p_lo                 # minimise the arm's NMB value-to-go
  } else {
    if (coef < 0) p_lo else p_hi                 # maximise
  }
  Vt[1] <- (1 - d0) * Vnext[1] + d0 * Vnext[3] + p * coef
  Vt
}

# Backward recursion over reward steps t = horizon-1, ..., 0 with terminal V = 0.
arm_envelope <- function(p_lo, p_hi, sense, closure) {
  V <- c(0, 0, 0)
  for (t in (horizon - 1):0) {
    dt <- 1 / (1 + discount_rate)^t
    g  <- dt * (lambda * u - c_state)            # state NMB contribution g_t
    V  <- g + step_value(V, p_lo, p_hi, dt, sense)
  }
  V[1] - if (closure) c_closure_once else 0      # m0 = e_1; subtract one-time closure cost
}

# IDM bounds (alpha = 2) from the pooled high-risk PFO evidence.
alpha          <- 2
closure_py     <- 238 * 5.3 + 60 * 2.8
medical_py     <- 235 * 5.3 + 60 * 2.8
closure_events <- 0
medical_events <- 20
pc_lo <- closure_events / (closure_py + alpha)
pc_hi <- (closure_events + alpha) / (closure_py + alpha)
pm_lo <- medical_events / (medical_py + alpha)
pm_hi <- (medical_events + alpha) / (medical_py + alpha)

# Independent per-arm admissible sets, so the lower endpoint of INMB factorises:
#   lower INMB = inf NMB_closure - sup NMB_medical   (closure worst, medical best)
#   upper INMB = sup NMB_closure - inf NMB_medical   (closure best, medical worst)
lower_INMB <- arm_envelope(pc_lo, pc_hi, "lower", TRUE) -
              arm_envelope(pm_lo, pm_hi, "upper", FALSE)
upper_INMB <- arm_envelope(pc_lo, pc_hi, "upper", TRUE) -
              arm_envelope(pm_lo, pm_hi, "lower", FALSE)

cat(sprintf("Lower INMB (Theorem 1 recursion): %.1f\n", lower_INMB))
cat(sprintf("Upper INMB (Theorem 1 recursion): %.1f\n", upper_INMB))
cat(sprintf("IDM INMB interval: [%.0f, %.0f]\n", round(lower_INMB), round(upper_INMB)))
stopifnot(lower_INMB < 0, upper_INMB > 0)   # interval crosses zero: non-robust
\end{lstlisting}

\hypertarget{case_study_reproducible.r}{%
\subsection{\texorpdfstring{\texttt{case\_study\_reproducible.R}}{case\_study\_reproducible.R}}\label{case_study_reproducible.r}}

Reproduces the PFO case study: evidence synthesis, IDM bounds, the
cohort model, the results table, and the threshold/strength sensitivity
analyses.

\begin{lstlisting}
# Real-data case study: PFO closure after cryptogenic stroke
# This script reproduces the PFO case study (Real-data application) in the manuscript.
# It uses a compact Markov cohort model to isolate transition-row uncertainty.

# -------------------------
# Evidence synthesis
# -------------------------
# High-risk PFO recurrent-stroke evidence used in the manuscript:
# CLOSE: 0/238 closure vs 14/235 antiplatelet, mean/median follow-up approximated as 5.3 time steps.
# DEFENSE-PFO: 0/60 closure vs 10.5% medical therapy, represented as 6/60, follow-up 2.8 time steps.
# Patient-years are approximated as randomized patients multiplied by follow-up.

trials <- data.frame(
  trial = c("CLOSE", "CLOSE", "DEFENSE-PFO", "DEFENSE-PFO"),
  strategy = c("PFO closure", "Medical therapy", "PFO closure", "Medical therapy"),
  patients = c(238, 235, 60, 60),
  followup = c(5.3, 5.3, 2.8, 2.8),
  stroke_events = c(0, 14, 0, 6)
)

trials$patient_years <- trials$patients * trials$followup
print(trials)

closure_py <- sum(trials$patient_years[trials$strategy == "PFO closure"])
medical_py <- sum(trials$patient_years[trials$strategy == "Medical therapy"])
closure_events <- sum(trials$stroke_events[trials$strategy == "PFO closure"])
medical_events <- sum(trials$stroke_events[trials$strategy == "Medical therapy"])

alpha <- 2

# Empirical annual recurrent-stroke probabilities.
p_closure_emp <- closure_events / closure_py
p_medical_emp <- medical_events / medical_py

# Precise symmetric Dirichlet posterior mean over the binary row: stroke / no stroke.
p_closure_dir <- (closure_events + alpha / 2) / (closure_py + alpha)
p_medical_dir <- (medical_events + alpha / 2) / (medical_py + alpha)

# IDM bounds for the recurrent-stroke component.
p_closure_lower <- closure_events / (closure_py + alpha)
p_closure_upper <- (closure_events + alpha) / (closure_py + alpha)
p_medical_lower <- medical_events / (medical_py + alpha)
p_medical_upper <- (medical_events + alpha) / (medical_py + alpha)

cat("\nAnnual recurrent-stroke probabilities\n")
print(data.frame(
  strategy = c("PFO closure", "Medical therapy"),
  empirical = c(p_closure_emp, p_medical_emp),
  dirichlet_mean = c(p_closure_dir, p_medical_dir),
  IDM_lower = c(p_closure_lower, p_medical_lower),
  IDM_upper = c(p_closure_upper, p_medical_upper)
))

# -------------------------
# Markov cohort model
# -------------------------
# States: no recurrent stroke, after recurrent stroke, dead.

evaluate_strategy <- function(p_stroke,
                              closure = FALSE,
                              horizon = 7,
                              lambda = 150000,
                              discount_rate = 0.03,
                              p_death_no_stroke = 0.002,
                              p_death_post_stroke = 0.050,
                              utility_no_stroke = 0.86,
                              utility_post_stroke = 0.65,
                              cost_no_stroke = 200,
                              cost_post_stroke = 8000,
                              cost_acute_stroke = 40000,
                              cost_closure_once = 15000) {
  m <- c(no_stroke = 1, post_stroke = 0, dead = 0)
  qalys <- 0
  costs <- if (closure) cost_closure_once else 0

  for (t in 0:(horizon - 1)) {
    d <- 1 / (1 + discount_rate)^t

    # Time-step outcome contributions.
    qalys <- qalys + d * (m["no_stroke"] * utility_no_stroke +
                          m["post_stroke"] * utility_post_stroke)
    costs <- costs + d * (m["no_stroke"] * cost_no_stroke +
                          m["post_stroke"] * cost_post_stroke)

    # Acute recurrent stroke cost for transitions from no recurrent stroke to post-stroke.
    strokes <- m["no_stroke"] * p_stroke
    costs <- costs + d * strokes * cost_acute_stroke

    P <- matrix(
      c(1 - p_stroke - p_death_no_stroke, p_stroke, p_death_no_stroke,
        0, 1 - p_death_post_stroke, p_death_post_stroke,
        0, 0, 1),
      nrow = 3,
      byrow = TRUE,
      dimnames = list(names(m), names(m))
    )

    m <- as.numeric(m %*% P)
    names(m) <- c("no_stroke", "post_stroke", "dead")
  }

  nmb <- lambda * qalys - costs
  list(qalys = qalys, costs = costs, nmb = nmb)
}

compare_strategies <- function(p_closure, p_medical) {
  closure <- evaluate_strategy(p_closure, closure = TRUE)
  medical <- evaluate_strategy(p_medical, closure = FALSE)
  c(
    closure_qalys = closure$qalys,
    closure_costs = closure$costs,
    closure_nmb = closure$nmb,
    medical_qalys = medical$qalys,
    medical_costs = medical$costs,
    medical_nmb = medical$nmb,
    incremental_qalys = closure$qalys - medical$qalys,
    incremental_costs = closure$costs - medical$costs,
    INMB = closure$nmb - medical$nmb
  )
}

results <- rbind(
  empirical_matrix = compare_strategies(p_closure_emp, p_medical_emp),
  precise_dirichlet_mean = compare_strategies(p_closure_dir, p_medical_dir),
  IDM_lower_bound = compare_strategies(p_closure_upper, p_medical_lower),
  IDM_upper_bound = compare_strategies(p_closure_lower, p_medical_upper)
)

cat("\nCost-effectiveness results\n")
print(round(results, 3))

# -------------------------
# Sensitivity analyses
# (Section: Sensitivity to the threshold and the imprecision strength)
# -------------------------

# INMB at a given IDM strength alpha and threshold lambda, as the exact lower/upper
# envelope: closure worst (highest stroke prob) and medical best (lowest) give the
# lower endpoint; the reverse gives the upper endpoint.
inmb_interval <- function(alpha, lambda = 150000) {
  pc_lo <- closure_events / (closure_py + alpha)
  pc_hi <- (closure_events + alpha) / (closure_py + alpha)
  pm_lo <- medical_events / (medical_py + alpha)
  pm_hi <- (medical_events + alpha) / (medical_py + alpha)
  lower <- evaluate_strategy(pc_hi, closure = TRUE,  lambda = lambda)$nmb -
           evaluate_strategy(pm_lo, closure = FALSE, lambda = lambda)$nmb
  upper <- evaluate_strategy(pc_lo, closure = TRUE,  lambda = lambda)$nmb -
           evaluate_strategy(pm_hi, closure = FALSE, lambda = lambda)$nmb
  c(lower = lower, upper = upper)
}

cat("\nINMB interval at lambda = 150000 by IDM strength alpha\n")
for (a in c(1, 2)) {
  iv <- inmb_interval(a, 150000)
  cat(sprintf("  alpha = %d: [%.0f, %.0f]\n", a, iv["lower"], iv["upper"]))
}

cat("\nINMB interval by willingness-to-pay threshold (alpha = 2)\n")
for (lam in c(50000, 100000, 150000, 200000)) {
  iv <- inmb_interval(2, lam)
  concl <- if (iv["lower"] > 0) "PFO closure" else
           if (iv["upper"] < 0) "Medical therapy" else "None (non-robust)"
  cat(sprintf("  lambda = %6d: [%.0f, %.0f]  -> %s\n", lam, iv["lower"], iv["upper"], concl))
}

# Threshold at which each curve crosses zero (alpha = 2):
# INMB(lambda) = lambda * dQALY - dCost, so the crossing is at lambda = dCost / dQALY.
zero_crossing <- function(p_closure, p_medical) {
  cl <- evaluate_strategy(p_closure, closure = TRUE)
  me <- evaluate_strategy(p_medical, closure = FALSE)
  (cl$costs - me$costs) / (cl$qalys - me$qalys)
}
pc_lo2 <- closure_events / (closure_py + 2)
pc_hi2 <- (closure_events + 2) / (closure_py + 2)
pm_lo2 <- medical_events / (medical_py + 2)
pm_hi2 <- (medical_events + 2) / (medical_py + 2)
cat(sprintf("\nZero-crossing thresholds (alpha = 2):\n  empirical:      %.0f\n  lower envelope: %.0f\n  upper envelope: %.0f\n",
            zero_crossing(p_closure_emp, p_medical_emp),
            zero_crossing(pc_hi2, pm_lo2),
            zero_crossing(pc_lo2, pm_hi2)))
\end{lstlisting}

\hypertarget{classical_bellman_three_state.r}{%
\subsection{\texorpdfstring{\texttt{classical\_bellman\_three\_state.R}}{classical\_bellman\_three\_state.R}}\label{classical_bellman_three_state.r}}

Minimal classical (precise) three-state backward Bellman recursion,
illustrating the singleton reduction (Corollary 1).

\begin{lstlisting}
# Classical three-state Markov cohort model using Bellman recursion
# This script demonstrates the classical backward Bellman recursion of the manuscript.
# States: S1 = healthy, S2 = sick, S3 = dead.

P <- matrix(
  c(0.80, 0.15, 0.05,
    0.00, 0.70, 0.30,
    0.00, 0.00, 1.00),
  nrow = 3,
  byrow = TRUE
)

m0 <- c(1, 0, 0)       # whole cohort starts in S1
T <- 3                 # final model year
g <- c(1, 1, 0)        # survival-year value vector

# Backward Bellman recursion.
# V[[t + 1]] stores V_t because R list indices start at 1.
V <- vector("list", T + 1)
V[[T + 1]] <- g

for (t in (T - 1):0) {
  V[[t + 1]] <- g + as.numeric(P %*% V[[t + 2]])
}

cat("Bellman value vectors:\n")
for (t in 0:T) {
  cat("V_", t, " = ", paste(round(V[[t + 1]], 4), collapse = ", "), "\n", sep = "")
}

cat("\nExpected accumulated value from m0:", sum(m0 * V[[1]]), "\n")

# Forward cohort recursion check.
m <- matrix(NA_real_, nrow = T + 1, ncol = 3)
m[1, ] <- m0

for (t in 0:(T - 1)) {
  m[t + 2, ] <- m[t + 1, ] %*% P
}

year_values <- as.numeric(m %*% g)

cat("\nForward cohort states:\n")
print(round(m, 4))
cat("Year values:", paste(round(year_values, 4), collapse = ", "), "\n")
cat("Forward accumulated value:", sum(year_values), "\n")
\end{lstlisting}

\hypertarget{simulation_sparsity.r}{%
\subsection{\texorpdfstring{\texttt{simulation\_sparsity.R}}{simulation\_sparsity.R}}\label{simulation_sparsity.r}}

Sparsity experiment: decision error and IDM abstention rates versus
per-arm patient-years (Figure 2).

\begin{lstlisting}
# simulation_sparsity.R
# Decision outcomes (correct / wrong / abstain) for the empirical, precise
# Dirichlet, and IDM analyses as the amount of transition evidence varies.
# Reproduces the "Behavior under increasing sparsity" subsection (Figure).
#
# Ground truth: PFO closure genuinely lowers the per-step recurrent-stroke
# probability from 0.014 (medical therapy) to 0.010 (closure), but not by
# enough to justify its one-time cost. The true incremental net monetary
# benefit is negative, so the correct decision is medical therapy. Under
# sparse evidence the closure arm frequently records a sampling zero, which
# the point analyses can misread as a highly effective treatment.

set.seed(20260621)

# Net monetary benefit of one arm (3-state cohort model of the case study).
nmb_arm <- function(p_stroke, closure = FALSE, horizon = 7,
                    lambda = 150000, discount_rate = 0.03,
                    p_death_no_stroke = 0.002, p_death_post_stroke = 0.050,
                    utility_no_stroke = 0.86, utility_post_stroke = 0.65,
                    cost_no_stroke = 200, cost_post_stroke = 8000,
                    cost_acute_stroke = 40000, cost_closure_once = 15000) {
  m <- c(1, 0, 0)
  qalys <- 0
  costs <- if (closure) cost_closure_once else 0
  for (t in 0:(horizon - 1)) {
    d <- 1 / (1 + discount_rate)^t
    qalys <- qalys + d * (m[1] * utility_no_stroke + m[2] * utility_post_stroke)
    costs <- costs + d * (m[1] * cost_no_stroke + m[2] * cost_post_stroke) +
             d * m[1] * p_stroke * cost_acute_stroke
    P <- matrix(c(1 - p_stroke - p_death_no_stroke, p_stroke, p_death_no_stroke,
                  0, 1 - p_death_post_stroke, p_death_post_stroke,
                  0, 0, 1), nrow = 3, byrow = TRUE)
    m <- as.numeric(m %*% P)
  }
  lambda * qalys - costs
}

inmb <- function(p_closure, p_medical, lambda = 150000) {
  nmb_arm(p_closure, TRUE, lambda = lambda) - nmb_arm(p_medical, FALSE, lambda = lambda)
}

pc_true <- 0.010   # true per-step recurrent-stroke probability, closure arm
pm_true <- 0.014   # true per-step recurrent-stroke probability, medical arm
lambda  <- 150000
alpha   <- 2
true_decision <- if (inmb(pc_true, pm_true, lambda) > 0) "closure" else "medical"
cat(sprintf("True INMB = %.0f  ->  correct decision: %s\n",
            inmb(pc_true, pm_true, lambda), true_decision))

Ns   <- c(25, 50, 100, 200, 400, 800, 1600, 3200)  # patient-years at risk per arm
reps <- 16000

simulate_N <- function(N) {
  emp_wrong <- dir_wrong <- idm_wrong <- idm_abstain <- idm_correct <- 0
  for (r in seq_len(reps)) {
    ec <- rpois(1, N * pc_true)   # recurrent strokes observed, closure arm
    em <- rpois(1, N * pm_true)   # recurrent strokes observed, medical arm

    # Empirical transition matrix.
    de <- if (inmb(ec / N, em / N, lambda) > 0) "closure" else "medical"

    # Precise symmetric Dirichlet posterior mean (strength alpha).
    pcd <- (ec + alpha / 2) / (N + alpha)
    pmd <- (em + alpha / 2) / (N + alpha)
    dd <- if (inmb(pcd, pmd, lambda) > 0) "closure" else "medical"

    # IDM interval: closure worst (upper p) vs medical best (lower p) gives the
    # lower INMB endpoint; the reverse gives the upper endpoint.
    lo <- nmb_arm((ec + alpha) / (N + alpha), TRUE,  lambda = lambda) -
          nmb_arm( em        / (N + alpha), FALSE, lambda = lambda)
    hi <- nmb_arm( ec        / (N + alpha), TRUE,  lambda = lambda) -
          nmb_arm((em + alpha) / (N + alpha), FALSE, lambda = lambda)
    di <- if (lo > 0) "closure" else if (hi < 0) "medical" else "abstain"

    emp_wrong <- emp_wrong + (de != true_decision)
    dir_wrong <- dir_wrong + (dd != true_decision)
    if (di == "abstain") idm_abstain <- idm_abstain + 1
    else if (di == true_decision) idm_correct <- idm_correct + 1
    else idm_wrong <- idm_wrong + 1
  }
  c(N = N,
    emp_wrong   = emp_wrong   / reps,
    dir_wrong   = dir_wrong   / reps,
    idm_wrong   = idm_wrong   / reps,
    idm_abstain = idm_abstain / reps,
    idm_correct = idm_correct / reps)
}

results <- as.data.frame(t(sapply(Ns, simulate_N)))
cat("\nDecision outcomes by patient-years per arm (alpha = 2, lambda = 150000)\n")
print(round(results, 3))

# These Monte Carlo estimates reproduce Figure "Decision outcomes under
# increasing evidence" up to simulation error.
\end{lstlisting}